\documentclass[final,3p,times, authoryear]{elsarticle}

\usepackage{amssymb}
\usepackage{amsmath}
\usepackage{amsthm}
\usepackage{float}
\usepackage{booktabs}

\usepackage[most]{tcolorbox}
\usepackage{xcolor}
\usepackage{subfig}
\usepackage{graphicx}
\usepackage{siunitx} 

\usepackage{etoolbox}
\makeatletter
\patchcmd{\subsection}{\normalfont}{\normalfont\bfseries}{}{}
\patchcmd{\subsubsection}{\normalfont}{\normalfont\bfseries}{}{}
\makeatother

\theoremstyle{definition}
\newtheorem{theorem}{Theorem}[section]
\newtheorem{definition}[theorem]{Definition}
\newtheorem{remark}[theorem]{Remark}
\newtheorem{proposition}[theorem]{Proposition}
\newtheorem{example}[theorem]{Example}
\newtheorem{conjecture}[theorem]{Conjecture}  
\newtheorem{claim}[theorem]{Claim}
\newcommand{\gcdd}[1]{\mathrm{GCD(#1)}}
\newcommand{\bigA}{\mathcal{A}}
\newcommand{\bigD}{\mathbf{D}}

\newcommand{\bigT}{\mathbf{T}}

\newcommand{\sumQ}{\sum_{j=1}^Q}

\renewcommand{\cite}[1]{\textcolor{blue}{\textit{\citep{#1}}}}

\journal{International Symposium on Transportation and Traffic Theory 2026}

\begin{document}

\begin{frontmatter}



\title{Wardropian Cycles make traffic assignment both optimal and fair by eliminating price-of-anarchy with Cyclical User Equilibrium for compliant connected autonomous vehicles.}

 
\author[a]{Michał Hoffmann}
\author[a]{Michał Bujak}
\author[a]{Grzegorz Jamróz}
\author[a]{Rafał Kucharski}

\affiliation[a]{organization={Jagiellonian University},
            city={Kraków},
            country={Poland}}
\begin{abstract}



    Connected and Autonomous Vehicles (CAVs) open the possibility for centralised routing with full compliance, making System Optimal traffic assignment attainable. However, as System Optimum  makes some drivers better off than others, voluntary acceptance seems dubious. To overcome this issue, we propose a new concept of Wardropian cycles, which, in contrast to previous utopian visions, makes the assignment fair on top of being optimal, which amounts to satisfaction of both Wardrop’s principles.
    Such cycles, represented as sequences of permutations to the daily assignment matrices, always exist and equalise, after a limited number of days, average travel times among travellers (like in User Equilibrium) while preserving everyday optimality of path flows (like in System Optimum). We propose exact methods to compute such cycles and reduce their length and within-cycle inconvenience to the users. As identification of optimal cycles turns out to be NP-hard in many aspects, we introduce a greedy heuristic efficiently approximating the optimal solution. Finally, we introduce and discuss a new paradigm of Cyclical User Equilibrium, which ensures stability of optimal Wardropian Cycles under unilateral deviations.  
    
    We complement our theoretical study with large-scale simulations. In Barcelona, 670 vehicle-hours of Price-of-Anarchy are eliminated using cycles with a median length of 11 days—though 5\% of cycles exceed 90 days. However, in Berlin, just five days of applying the greedy assignment rule significantly reduces initial inequity. In Barcelona, Anaheim, and Sioux Falls, less than 7\% of the initial inequity remains after 10 days, demonstrating the effectiveness of this approach in improving traffic performance with more ubiquitous social acceptability.
   
\end{abstract}

\end{frontmatter}



\section{Introduction}
With increasing environmental awareness, societies demand to reduce traffic externalities, both by limiting traffic volumes and improving system efficiency \cite{lindsey2020addressing}. Variety of push-and-pull measures is applied worldwide with aim to reduce CO2 emissions, noise, fuel consumption, etc.  \cite{hrelja2023decreasing}. With the advent of connected autonomous vehicles (CAVs), policymakers anticipate their benefits in traffic flow efficiency, increased safety, and lower externalities \cite{pribyl2020addressing}. One of the potential benefits of autonomous driving lies in collective coordinated routing, where vehicles are centrally assigned to routes to meet system-wide objectives \cite{cantarella2017transportation, wang2019multiclass}. 
Such opportunity naturally corresponds to the classic notion of system optimal (SO) assignment, where total system costs are minimised \cite{Wardrop1952}, which, at the same time, is the explicit goal of sustainable transportation policies. 
Traditionally, though, system optimal assignments, despite their potential to improve system performance and reduce externalities, were deemed an unattainable theoretical concept without practical application. Thanks to technological advancements of CAVs, coupled with the conceptual framework proposed here, this may change. 

In principle, there are two classes of solutions to traffic assignment problems. 
First, rational humans arrive at the User Equilibrium (UE) where none of the drivers can individually benefit from changing their chosen routes \cite{Wardrop1952}. 
It is widely accepted that such an equilibrium arises in natural conditions when drivers, rational utility maximisers, do not coordinate their decisions and aim to maximise personal utility (or minimise perceived travel time) \cite{morandi2024bridging}. 
UE typically leads to suboptimal system performance (typically measured with total travel times or costs), which can be improved.
The system performance, on the other hand, is best in the System Optimal (SO) assignment, which yields minimal total travel costs. 

The difference between the system performance in SO and UE, measured with the so-called Price-of-Anarchy (PoA) \cite{youn2008price}, tells how much users could benefit (globally) when coordinating towards global optimum.
PoA was perceived indispensable in social systems with individual, rational, not cooperating utility maximisers, and this will not change just because drivers will switch to CAVs. Humans, regardless of whether they are driving or being driven in CAV, will demand to arrive as fast as possible and will remain reluctant to independently pay the costs of system-wide improvements. Thus, two Wardrop principles have so far been considered irreconcilable: either the equity of uncoordinated rational utility maximisers yields suboptimal User Equilibrium, or fully compliant drivers following centralised orders collectively reach a global goal of System Optimum, regardless of individual costs.

Previous system-optimal assignment concepts proposed for CAVs \cite{wang2018so} rely on the assumption of full compliance of individuals to reach global optimum (minimising times or emissions), with the individual driver perspective typically outside of scope. 
We argue this will be not acceptable by modern liberal societies, due to inherent inequalities of system optimal solution \cite{friedman1962capitalism}. Optimal solutions, while minimising the system-wide costs, neglect the costs of individuals, some of whom may experience (and do experience, as we show in our results) longer travel times (greater travel costs) compared to their colleagues (arriving to the same destination faster). As in other similar contexts, the system-wide benefits will not encourage individuals to sacrifice their own utility \cite{kahneman1984choices}. 
Free societies are likely to rebel, and individuals will opt out from any unfair system, jeopardising the potential benefits, reverting the system back to classical user equilibrium. To comply with system optimum, drivers would need either external incentives (e.g. mobility credit \cite{nie2012transaction} or pricing \cite{rambha2018marginal}), or a dictator who will reduce social score for non-compliance \cite{li2025cooperative}, or a new different class of solutions.
To exploit the arising opportunity, a new paradigm and theory bridging between the two concepts of network assignment is needed.

We show that, with a slight change of perspective, the fairness can be preserved at system optimum.
We depart from two assumptions. The first is that the gains from coordination shall be not only global, but also individual and fairly distributed.
Second, CAV users, if properly informed, should be happy to accept day-to-day variability if their average travel times are reduced in the long term. 
This requires, however, shifting the perspective from a single day of commute (as typically in traffic assignment context), to the series of days. To reach the system optimum on a given day, drivers need to coordinate (e.g. via centralised CAV fleet dispatcher), yet to make the assignment fair, this coordination needs to span over several days, like in the following example:



\begin{tcolorbox}[colback=gray!10,colframe=gray!30, coltitle=black, title = Wardopian Cycle - illustrative example]
Three drivers use capacitated routes to get from a common origin to a common destination departing at the same time. Uncoordinated, each of them travels 5 minutes (15 minutes in total).
Suppose that, if coordinated, it can be reduced to an average of 4 minutes (12 minutes in total), which will yield lower externalities (fuel, noise, etc.). This, despite coordination, requires one driver to travel 6 minutes instead of 3, like the remaining two. 
Why would the driver travelling 6 minutes agree to that? This clearly violates the Wardrop first principle.

Here, the solution is a 3-day cycle: \emph{everyone for two consecutive days enjoys a faster route, followed with a one-day of travel on a longer route}. This way, the assignment is both optimal (average travel time is 4 minutes every day) and fair (it is equal among travellers after 3 days).
\end{tcolorbox}

In the above scheme, everyday drivers optimally load the network with a pattern that yields minimal costs. By following system-optimal assignment drivers on each day will experience different travel times and the burden of reducing system-wide costs will be distributed unequally: some drivers will arrive earlier, some later. This can be alleviated by how drivers are assigned to routes every day.
We call such a sequence of daily assignments a \emph{Wardopian Cycle}, demonstrate that it always exists and is finite. Drivers following the daily assignments after completing the Wardopian Cycle simultaneously experience: i) the same average travel times, ii) the system-optimal travel times (typically faster than User Equilibrium). This constitutes a new concept of traffic equilibrium which is both fair and optimal, denoted Cyclical User Equilibrium (CUE).

We argue that, if transparently demonstrated, fleet assignment systems based on our proposed Wardopian Cycles can become a widely accepted and adopted solution for future traffic with CAVs. Unlike in purely centralised settings, the system-wide performance under Wardopian Cycles can be improved with acceptance of individuals (who can still travel faster). The outcomes, daily monitored and explicitly reported to the drivers, demonstrating reduced travel times over the selfish user equilibrium baseline, would motivate users to remain in the system. With an increasing share of drivers registered to the system, it will scale, leading to increased savings,  
which may be then allocated to pursue non-individual goals (like emissions or noise). Operationalisation of the theoretical concepts introduced in this study can lead to improved traffic performance without sacrificed equity. 

This paper contributes with:
\begin{itemize}
\item demonstrating that the traffic assignments can remain fair at system optimum, which is achieved if the coordination between users spans beyond a single day,
\item showing that the Price-of-Anarchy can be eliminated in a socially acceptable manner, which creates a new postulated concept of traffic equilibrium: Cyclical User Equilibrium,
\item proving that travel time equalising cycles always exist and proposing series of methods to reduce their lengths and users' discontent,
\item generalising Wardropian Cycles with the concept of Inequity Minimising Daily Assignment, applicable in less strict conditions,
\item demonstrating the results in the real-case of Barcelona, replicated on four more cities,
\item opening the new avenue of research, hopefully leading to a practical CAV fleet assignment paradigms.
\end{itemize}



\begin{figure}[ht!]
\begin{center}
\includegraphics[width=\textwidth]{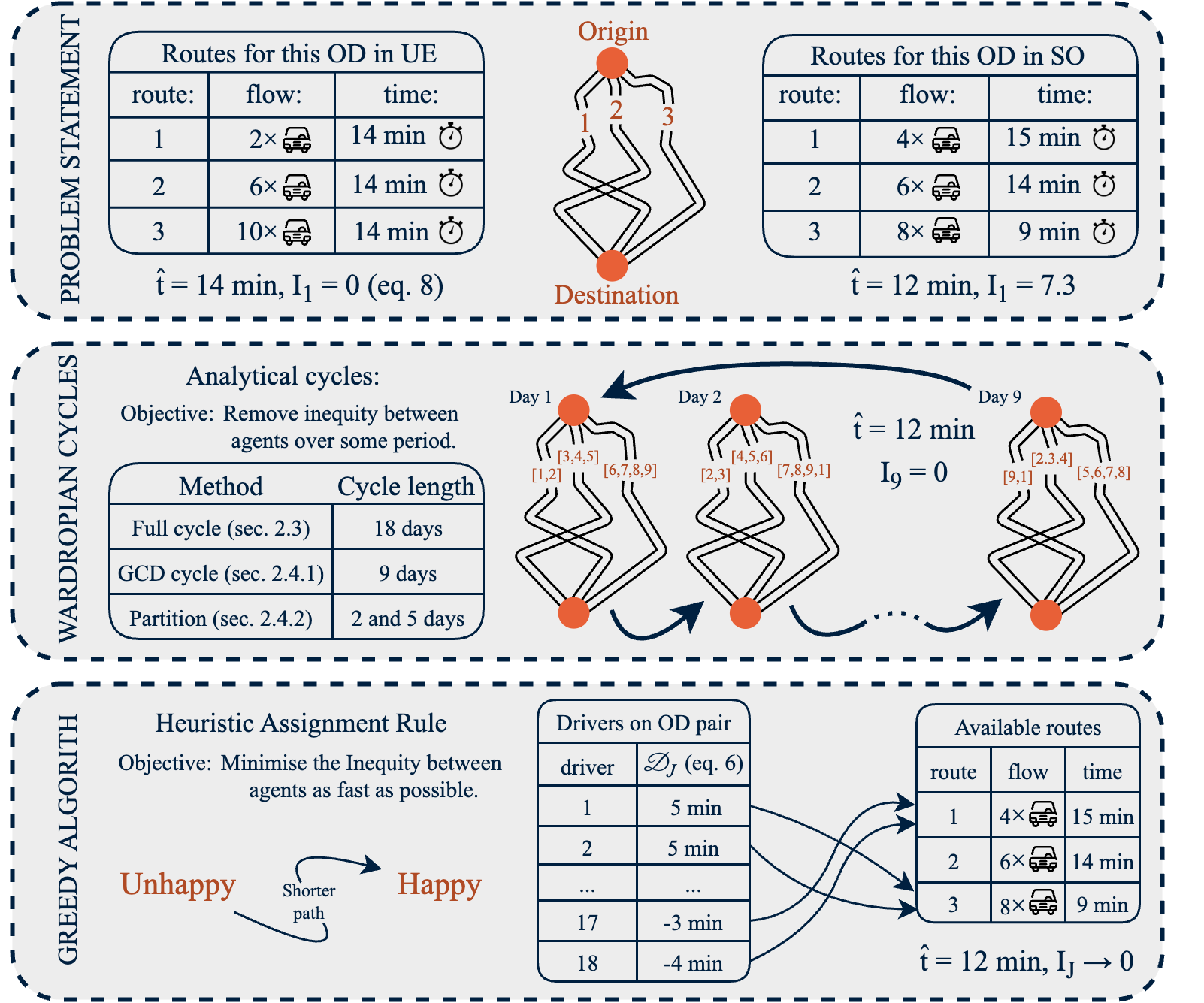}
\end{center}
\caption{
Wardropian Cycles concept at glance: \textbf{1.} In comparison to User Equilibrium, System Optimal assignment can be faster, but unfair across users. We measure this unfairness by Inequity (eq. \ref{eq_inequity}). \textbf{2.} To achieve fairness in SO, for each OD pair, we create a \emph{Wardrop Cyclical} assignment rule to achieve equal average travel times after couple days, which we call a \textit{Wardropian Cycle}. We show how can they be improved (shortened, equalised and more easily accepted by participants). \textbf{3.} A greedy daily assignment rule assigns better off drivers to longer routes day by day, decreases inequity of the system quicker, equalising the average travel times of agents on one OD pair.}
\label{fig:methodology}
\end{figure}





\subsection{Background} 
User Equilibrium (UE) and System Optimum(SO) \cite{Wardrop1952} are the two most important principles of traffic assignment. The former (UE) arises as a \emph{equilibrium} assignment similar to the Nash equilibrium \cite{Nash1951}, where no user can improve their travel time, while the latter, \emph{not being an equilibrium}, makes the system-wide cost minimal; this, however, for the price of inequity of travel times on different routes, which is inherently unfair. Since \cite{Wardrop1952}, it has been known that the system cost of UE may be higher than the optimal system cost of SO. As this is rather unsatisfactory, multiple attempts have been made, see \cite{chow2007trip, van2016user,  morandi2024bridging} for recent reviews, to reduce this difference, nowadays called \emph{price of anarchy} \cite{koutsoupias1999worst}, which could reach up to $30\%$ in modern networks \cite{youn2008price}, compare also \cite{roughgarden2004bounding, correa2005inefficiency} for theoretical bounds. 
Following \cite{morandi2024bridging}, we distinguish two main approaches to bridge the gap between UE and SO: 1) identification of SO with given user constraints and 2) relaxation of UE. 

In the former approach Constrained System Optima (CSO) are considered which are optima with some constraints on available paths, e.g. forbidden very long paths or paths whose travel times deviate excessively from UE times \cite{Jahn2000, Jahn2005SystemOptimal}. 
Relaxing the UE conditions, on the other hand, may consist in finding the Constrained User Equilibrium (CUE), e.g. \cite{zhou2012user}, where only routes short enough in terms of euclidean length are considered, or using a Bounded Rational User Equilibrium, proposed in \cite{mahmassani1987bounded}, in which the strict equality of travel times for different routes connecting a given OD pair needs to be satisfied only up to a small tolerance band, which enlarges the solution set and potentially allows for more efficient system-wide assignments; see e.g. \cite{guo2022managing} for a recent study.
Another approach involves applying Nash Welfare \cite{lujak2015route}, where the trade-off between fairness and efficiency, see also \cite{bertsimas2012efficiency} for a general discussion, is precisely quantified, see also \cite{CHELLAPILLA2023384} for a recent general bi-objective optimisation framework aimed at reducing congestion. 

A direct approach to reducing the price of anarchy, while preserving most individuals' freedom to choose paths was proposed in \cite{Stackelberg_in_practice}, where Stackelberg routing applied by only a fraction of vehicles, which predicts other users' selfish behaviour \cite{krichene2014stackelberg} brings the system closer to SO, see also \cite{roughgarden2001stackelberg,CorreaStackelberg, yang2007stackelberg, li2022differentiable} for theoretical bounds and algorithms. 
The downside of SO-driven Stackelberg routing, however, is that it requires a fraction of vehicles to \emph{always} use the less favourable routes, which is not feasible in general. 



The above-mentioned studies assumed that combining the optimum system performance of SO while preserving the fairness and individual equilibration characteristic for UE is impossible save for serendipitous network and demand configurations when UE indeed coincides with SO. This being true in the one-day perspective, it is no longer, as we show in this paper, valid in the multi-day perspective. Still, to the best of our knowledge, the idea of reducing inequity over a number of days while maintaining exact System Optimum is novel.
Note that the game-theoretical mixed-strategies e.g. \cite{mixedstrategies} applied by individual drivers share some similarities with cycles that we introduce, e.g. the average number of drivers on every route could correspond to System Optimum. Nevertheless, mixed strategies are not only hard to implement in practice, but also exhibit residual random deviations from system optimal assignment, in contrast to cyclical assignments. 

Let us here point out that our framework refines SO via multi-day assignments, and so the problem of excessively long routes in SO persists.  Nevertheless,  our multi-day cyclical setting can be adapted to refine and make fair \emph{any} initial assignment, e.g., CSO in place of SO if this is more desirable. 
Another issue, raised in the literature, is how to make the users stick to the guidance as by deviating on a single day, a user of our framework may improve its travel time. Possible measures include using sensors and monitoring to detect unruly users and punish them, for example, by expulsion from the system \cite{lujak2015route}. Note that in our framework, the disciplinary measures are only optional, and even without them compliance could reach $100\%$ as the construction of our guidance potentially benefits \emph{all} users compared to UE. 
Settings when every journey incurs a cost via \emph{compulsory} congestion pricing are considered in \cite{verhoef2002second, de2011traffic} and may achieve system optimum even in dynamic settings \cite{vickrey1969congestion}, the pricing scenarios, however, have many drawbacks and we steer clear of them towards simple to implement, although multi-day guidance systems.

The general idea of \emph{guiding} (a fraction of) the drivers such that the system is closer to System Optimum is not new and dates back to the 1970s \cite{kobayashi1979feasibility}, see also, e.g. \cite{harker1988multiple, van1990combined, van1991multiple} for guidance systems with multiple user classes. However, to our knowledge, this guidance was always considered as a one-day scenario and also typically involved a fraction of users, see, e.g. \cite{zhang2018mitigating}, who are happy to follow the central guidance even though on an individual level they are disadvantaged. In contrast, in this paper we seek to establish a \emph{multi-day} guidance pattern which preserves the system optimum and which is \emph{advantageous} to \emph{all} the users.

Wardrop User Equilibrium based on selfish users minimising their travel times is by far not the only viable selfish assignment rule and can be conceived as a first approximation of reality. When the unrealistic assumption of perfect knowledge of all travel costs is relaxed, or users have different individual tastes regarding different routes, the Stochastic User Equilibrium \cite{daganzo1977stochastic} naturally arises. While we do not explicitly consider SUE in this paper, a framework based on preserving system optimum while optimising individual utilities in a multi-day horizon is a natural extension of our framework to consider in the future. In this context, let us also mention Karma, e.g. \cite{elokda2025carma, riehl2024karma}, where individually varying day-to-day preferences are taken into account. We note, however, that Karma primarily aligns with users' urgencies, and whether it reaches SO depends on fine-tuning of parameters of Karma auctions, which could be challenging as it depends on population profile. In comparison, our framework \emph{enforces} System Optimum without any other prerequisites, disregarding the day-to-day variation of urgencies.  Still, even if all users have the same urgencies, the Karma assignment remains stochastic in contrast to our precisely tailored system optimal assignment aimed at reducing inequity and boosting voluntary compliance.   


Even though the route guidance system we propose is centralised, which has some drawbacks like communication challenges and privacy issues, see \cite {khanjary2012route}, compared to decentralised systems such as \cite{lujak2015route, riehl2024karma}, the routing decisions are taken only infrequently, for example, once per day in the case of a morning commute. Consequently, the communications problems are easy to cope with as compared with dynamic routing settings,  e.g. \cite{liang2014real}. Privacy does not have to be an issue in our framework either as the only piece of information that could be revealed is the route taken, in contrast to, e.g., daily preferences discussed for Karma \cite{riehl2024karma}.  




Finally, let us point out that even though our framework is implementable as a multi-day route guidance system without any particular connectivity between the drivers (indeed, every user needs only to update the route at most once per day or even download the whole routing schedule once every, say, month), it is particularly easy to put in practice in the Connected (not necessarily autonomously driving) Vehicles setting \cite{Bagloee_AV_impact, Bagloee_mixedUE_SO}. Nevertheless, we {\bf do not} require all vehicles to be automated or even connected to benefit from the proposed multi-day routing pattern, and a system based on Wardropian Cycles in theory could be deployed in our cities even now.  


The paper is organised as follows. 
In Section 2, we formalise the concept of Wardropian Cycle and propose the methodology to identify short and attractive cycles, followed with more general concepts of Inequity minimising Daily Assignment, Cyclical User Equilibrium, and a new postulated class of OD-fair System Optimal Assignments. We complement the illustrative examples introduced along with the methodology with the results from Barcelona in Section 3, coupled with cross-city analysis. We conclude and open future research avenues in Section 4. The appendix contains more formal proofs, remarks, and illustrations.

\section{Methodology}\label{sec:methodology}

Consider a set of $Q$ drivers on an origin-destination (OD) pair, assigned to $K$ paths with a binary matrix $A \in \mathcal{M}^{Q \times K}(\{0,1\})$ (confer def. \ref{def:Daily_assignment}). 
Each driver is assigned to a path, and drivers assigned to path $k$ form a path flow $Q_k$.
In capacitated networks, path flows yield (through the traffic flow model) travel times $t_k(Q_k)$\footnote{While we use here the macroscopic $BPR$ functions, any model with a good estimate of travel times from path flows is applicable in the following (including, with some generalisation, the non-deterministic microscopic traffic flow models)}. 

\begin{definition}
\textbf{Assignment} maps $Q$ travellers to paths $1 \dots K$ forming respective flows $Q_k\geq0$ ($\sum_{k\leq K}  Q_k=Q$) which yields path travel times $t_k$: 
     \begin{equation}
     [Q_k, t_k]_{k \leq K} = \text{Assignment}(Q) 
     \label{eq:Assignment}
 \end{equation}
\end{definition}

We consider two specific assignments: 
\begin{enumerate}
    \item UE, where travel times on all paths are equal for all paths: $t_k=t_{k'} \, \forall k,k' \leq K$ and all users experience equal travel time $\hat{t}$:
     \begin{equation}
     [Q_k^{\text{UE}}, t_k^{\text{UE}}]_{k \leq K} = \text{UE}(Q) 
     \label{UE}
 \end{equation}
    \item SO, where the total travel times $\sum_k Q_k t_k$ are minimised (at the system, not on the OD pair level), but the travel times on paths and for users may differ. This yields the average travel time $\hat{t}=\sum_k Q_k t_k /Q$ typically shorter than the UE:
     \begin{equation}
     [Q_k^{\text{SO}}, t_k^{\text{SO}}]_{k \leq K} = \text{SO}(Q) 
     \label{SO}
 \end{equation}
\end{enumerate}

We introduce the third class:
\begin{enumerate}
\item[3] \emph{Wardrop Cyclical (WC)} which is defined as a sequence of daily assignment matrices $\bigA = (A_1, \ldots A_n)$ over $n$ consecutive days. Each of the daily assignment matrices $A_i$ preserves the desired path flows - in our case given by the SO assignment (eq. \ref{SO}).
\end{enumerate}

We argue and demonstrate (in Section \ref{sec:upper_boundary}) that for any OD pair and any assignment, there exists a finite sequence of daily assignments $\bigA$ such that the average travel time for each driver on this OD pair is equal.

\begin{definition}\label{def:Daily_assignment}
    Each Assignment has its feasible daily realizations that we denote a \textit{daily assignment}, which is a matrix $A_j \in \mathcal{M}^{Q \times K}(\{0,1\})$ which satisfies: 
\begin{subequations}
\begin{align}
&\sum_{k=1}^K [A_j]_{i,k} = 1 \mbox { for every } i \in \{1,\dots, Q\}.\label{eq_A1} \\
&\sum_{i=1}^Q [A_j]_{i,k} = Q_k \mbox{ for every } k \in \{1,\dots, K\}, \label{eq_A2}
\end{align}
\end{subequations}


$[A_j]_{i, k} = 1$ indicates that driver $i$ is assigned to route $k$ on day $j$. Condition \eqref{eq_A1} imposes that each driver is assigned by matrix $A$ to exactly one route and condition \eqref{eq_A2} states that the number of drivers on each route corresponds to a given assignment result $Q_k$. From now on, to simplify the notation, we assume that $Q_k$ is the SO solution $Q^{SO}_k$. 

\end{definition}

\subsection{Measures of inequity for non-UE assignments}

All assignments, apart from $\text{UE}$ are unfair, and travel times will vary between paths, each driver will experience their own travel time, according to the assigned path $t_{k}$, potentially deviating from the daily average on this OD pair $\hat{t}$. We quantify this with a series of measures: the vector $\bigD_j$ denotes the deviation on day $j$, it is cumulated until day $J$ with $\mathcal{D}_J$, and the progression of users discontent is denoted $\mathcal{H}^i_L$. We use second norm of deviation to measure the Inequity up to day $I_J$ (normalized in $\bar{I}_J$).

The overarching objective of \emph{Wardropian Cycles} is to eliminate or reduce the above inequity measures. We hypothesise that the cycles which quickly minimise some of those measures will be more easily accepted and thus widely adopted.
Verifying this, however, requires a separate study that either deduces compliance formulas from other studies and discrete choice models, or conducts a dedicated stated preference study.

\begin{definition}[Time Deviations]\label{def:time_deviations}
For a given day $j$ and daily assignment $A_j$, we define the vector of travel time deviations $D_j$ by

\begin{eqnarray}
\label{eq:time_deviations}
    \bigD_j 
    &=&  A_j  \begin{bmatrix} t_1 \\ t_2 \\ \dots \\ t_K \end{bmatrix} -  \begin{bmatrix}\hat{t} \\ \hat{t} \\ \dots \\ \hat{t}\end{bmatrix},
\end{eqnarray}
where $t_k$ is the path $k$ travel time and
 $\hat{t}$ is the average OD travel time.
\end{definition}

\begin{definition}[Cumulative Time Deviations]
\label{def:cum_deviations}
    To track the time deviation (def. \ref{def:time_deviations}) experienced by agents throughout some period $\{1,\dots,J\}$, we introduce the cumulative time deviations $\mathcal{D}_J$ defined as follows:
    \begin{equation}\label{eq:cum_time_dev}
        \mathcal{D}_J = \sum_{j=1}^J \bigD_j
    \end{equation}
\end{definition}

\begin{definition}[Discontent Progression]
\label{def:cum_discontent}
For every day up to $J$ of the daily assignments, we define the \emph{discontent progression} $\mathcal{H}_J$ as a matrix defined as follows:
\begin{equation}
\label{eq:history}
\mathcal{H}_J = [\mathcal{D}_1 \, \sum_{j=1}^2 \mathcal{D}_j \, \ldots \sum_{j=1}^J\mathcal{D}_J]
\end{equation}
\end{definition}

\begin{definition}[Inequity]
For a given sequence of cumulative deviations $\mathcal{D}_j$ up to a given day $j$, we define Inequity as a mean of squared deviations (where $||\cdot||_2$ is the $l_2$ norm) time deviations over all agents:
\begin{equation}
\label{eq_inequity}
    I_J = (||\mathcal{D}_J||_2)^2/Q.
\end{equation}

To make inequity system-wide comparable, we normalise it by dividing it by the mean travel time on OD pair (which allows to compare between different pairs, days, realisation of stochastic demand, travel times, etc). 
$\bar{I}$ as:
\begin{equation}
\label{eq:intequitynormalized}
    \bar{I}_J = I_J/\hat{t}
\end{equation}
\end{definition}


\subsection{Wardropian Cycle}

\begin{definition}\label{def:wardropian_cycle}
\textbf{Wardropian Cycle} of length $n$ is a finite sequence of daily assignments $\bigA = (A_1, \dots, A_n)$ such that $\mathcal{D}_n = 0$ (eq. \ref{eq:cum_time_dev}). 
\end{definition}

\begin{remark}
A Wardropian Cycle emerges when all travellers experience, over a fixed period of days, the same average travel time, equal to the average travel time on their OD pair. In particular, if the daily assignment $A_j$ is \emph{system optimal} every day,  travellers can reach better average time than in the User Equilibrium. 
The drivers, following daily assignments, are said to be in \emph{Wardrop Cyclical (WC)} assignment.
\end{remark}

The overview of methodology is depicted in Fig. \ref{fig:sect2-plan}.
In the following subsections, we first prove that Wardropian Cycles always exist and are finite.
Then, we show how to decrease their length, as a shorter cycle would make it easier to encourage drivers to follow assignments. 
Sections \ref{sub:opt_cycle} and \ref{sec:compliance} develop methods to describe and minimise the cost that individual agents pay by participating in a cycle. 
Later, in Section \ref{sec:greedy}, we describe an algorithm that leads to a more practical implementation of Wardopian cycles. The following part describes the switch from continuous traffic flow distributions to our discrete case.
In Section \ref{sec:new_wardrop}, we propose a new paraphrase of Wardrop principles and introduce a Cyclical Wardrop Equilibrium. 
We conclude with adding the new constraint to System Optimal assignments, compatible with Wardropian Cycles (in Sec. \ref{sec:CSO}) and generalising the concept to a \emph{bit} of traffic in Sec. \ref{sec:bit}.


\begin{figure}[ht!]
\begin{center}
\includegraphics[width=0.7\textwidth]{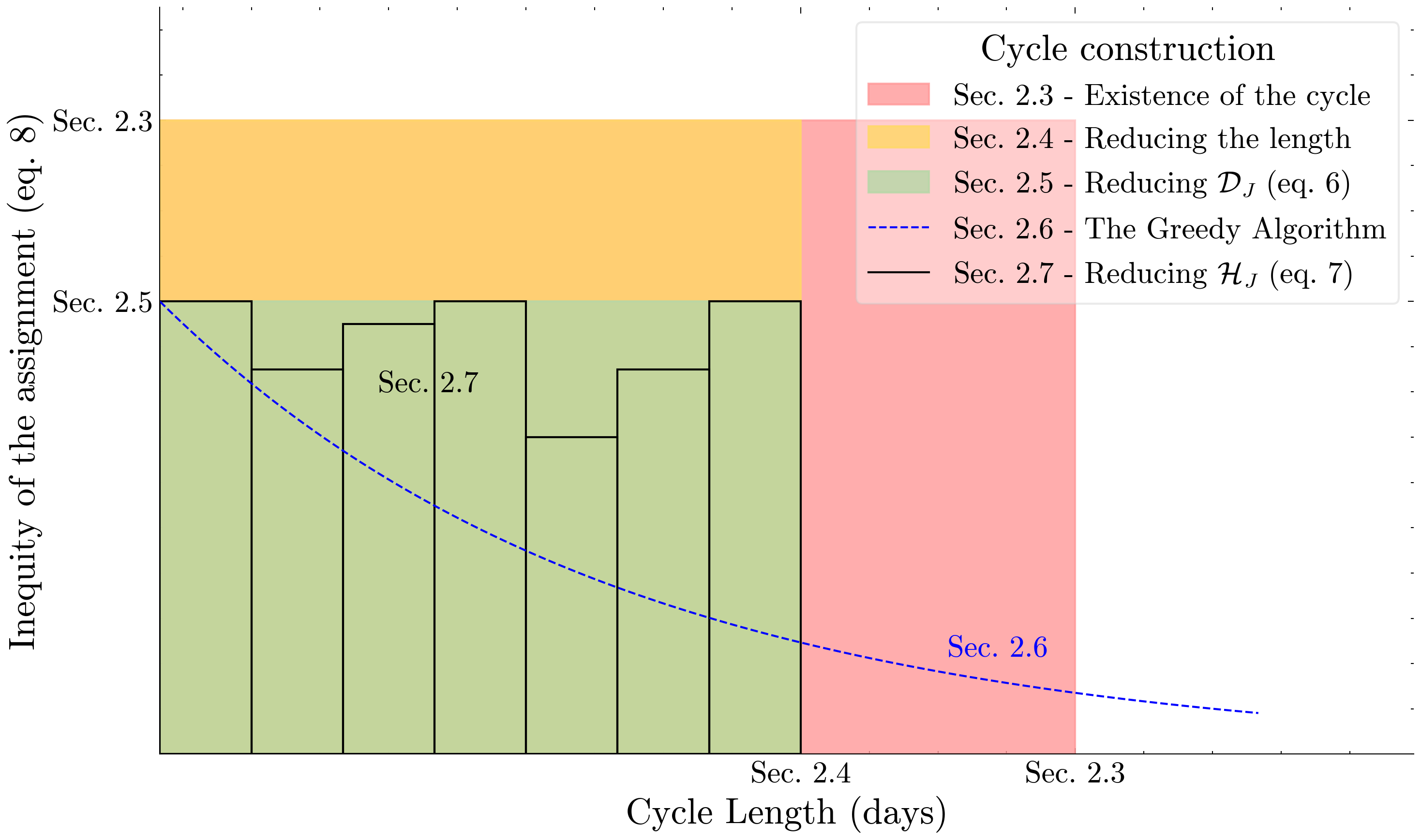}
\end{center}
\caption{\textbf{Methodology overview}: In Sec. 2.3 we prove existence of cycles, in 2.4 we show techniques to reduce their length, in 2.5 we focus on reducing the variance of travel times and generalise with greedy algorithm in Sec. 2.5. Methodology section concludes with an analysis of cumulated disutility and methods to mitigating its excess.}
\label{fig:sect2-plan}
\end{figure}

\subsection{Finite cycles always exist} 
\label{sec:upper_boundary}

For any OD pair and a given Assignment, a Wardropian Cycle always exists, we prove it with the existence of a Wardropian Cycle of length equal to the number of drivers $Q$ on this OD pair.  

\begin{definition}
    \textbf{Shift matrix} for an assignment $A$ is a permutation matrix $P$ defined as:
    $$P =
    \begin{bmatrix}
    0 & 1 & 0 & \dots & 0 & 0 \\
    0 & 0 & 1 & \dots & 0 & 0 \\
    \vdots & \vdots & \vdots & \ddots & \vdots & \vdots \\
    0 & 0 & 0 & \dots & 0 & 1 \\
    1 & 0 & 0 & \dots & 0 & 0
    \end{bmatrix},$$
    where number of rows and columns are equal to the number of columns and rows in $A$, respectively. 
\end{definition}

\begin{remark}
The permutation matrix $P$ shifts user assignment by one, i.e. given that travellers are assigned to routes $r_{i_1}, r_{i_2}, \ldots, r_{i_Q}$, on the following day they are assigned to routes $r_{i_2}, r_{i_3}, \ldots, r_{i_Q}, r_{i_1}$, respectively.
\end{remark}

\begin{proposition}
\label{Prop1}
For any initial assignment $A_1$, the sequence of daily assignments $A_1, A_2, \dots A_Q$ defined as
\begin{equation}
A_{j+1} = P^j A_1 \mbox{ for } j=1, \dots, Q-1,     
\end{equation}
where $P$ is a shift matrix for $A_1$, is a Wardropian Cycle of length $Q$.
\end{proposition}

\begin{proof}[Proof of Proposition \ref{Prop1}]
Let $T=[t_1, \ldots, t_K]^T$. The vector of total travel times of drivers can be expressed as
\label{proof_1}
$$\left(\sumQ A_j\right) T = 
\begin{bmatrix}
Q_1 & \dots & Q_K \\
 & \ddots &  \\
Q_1 & \dots & Q_K
\end{bmatrix} T =
\begin{bmatrix}
    \sum_{i=1}^K Q_i t_i \\
    \dots \\
    \sum_{i=1}^K Q_i t_i 
\end{bmatrix} = Q \begin{bmatrix}\hat{t} \\ \hat{t} \\ \dots \\ \hat{t}\end{bmatrix}. 
$$

After Q days, each driver will be assigned to route $k$ exactly $Q_k$ times, hence their mean travel time will be equal to $\hat{t}$.
\end{proof}



\subsection{Reducing the length of a cycle} 
\label{sec:GCD}
In this section, we propose a construction to make cycles shorter than $Q$ based on the greatest common divisor and a flow partition.

\subsubsection{Reduction by the greatest common divisor}
\label{sec:optim_flow}
The first method to reduce the length of a Wardropian Cycle is founded on the greatest common divisor (GCD). 
Instead of permuting the cycle by one shift daily ($\tau_1$), we use the GCD of path flows $M = \gcdd{Q_1, Q_2, \dots, Q_k}$ and shift drivers by this number of steps every day. We can do this by replacing the shift matrix $P$ in the cycle formulation with its $M$-th power: $P^M$.
\begin{proposition}
    For any initial assignment $A_1$ and $M = \gcdd{Q_1, Q_2, \dots, Q_k}$, we define daily assignments $A_2, \dots A_{Q/M}$ by
\begin{equation}
A_{j+1} = ( P^{M} )^j A_1 \mbox{ for } j=1, \dots, Q / M-1.      
\end{equation}
The sequence $A_1, \dots A_{Q/M}$ is a Wardropian Cycle of length $Q/M$.
\begin{proof}
Similarly to proof of Proposition \ref{Prop1}, we calculate the total travel time:
$$\left(\sum_{j=1}^{Q/M} A_j\right) T = 
\begin{bmatrix}
Q_1/M & \dots & Q_K/M \\
 & \ddots &  \\
Q_1/M & \dots & Q_K/M
\end{bmatrix} T =
\begin{bmatrix}
    \frac{1}{M}\sum_{i=1}^K Q_i t_i \\
    \dots \\
    \frac{1}{M}\sum_{i=1}^K Q_i t_i 
\end{bmatrix} = Q/M \begin{bmatrix}\hat{t} \\ \hat{t} \\ \dots \\ \hat{t}\end{bmatrix} 
$$
Instead of travelling each route $Q_k$ times, each driver is assigned to it exactly $Q_k/M$ times. Since, by the definition, for every $k$ this value is an integer, the resulting cycle is a Wardropian Cycle. 
\end{proof}

\end{proposition}
While $M=1$ brings us back to the solution from Section \ref{sec:upper_boundary}, often $M>1$ (as presented in Fig. \ref{fig:methodology}), which leads to significantly reduced cycle lengths.

\subsubsection{Reduction by flow partition} \label{sec:optim_times}
Further reducing a cycle length relies on splitting $Q$ drivers of each OD pair,  into groups $J_1, \dots, J_S$ of equal mean travel time. If such a partition is found, any Wardopian cycle of daily assignments can be applied within the group.

Now we no longer try to arrange every driver on the same cycle; instead, we create several, shorter ones that all have equal average time, but not necessarily equal length, nor route set. 
We further note that we can group routes with equal travel times together, as it is irrelevant whether a driver takes one route or another, the key factor being the travel time itself.
If so, we can assume that route times are pairwise different.
Therefore, instead of each route $r_k$ of flow $Q_k$, we have a set of $Q_k$ unit routes (routes of flow $1$) of equal time $t_k$. 
This leads to a multiset, where each route $r_k$ is present $Q_k$ times.
Hence, we can focus solely on travel times, while the flow information is incorporated in the multiset structure.
Such a representation allows us to divide a route flow into several sub-cycles, while preserving the SO assignment. Formally, we want to partition the multiset of size $Q$ into subsets $J_s$, where route $r_i$ (of flow $1$, present $Q_k$ times in the multiset) has time $t_i, i\leq Q$, so that

\begin{equation}
    \sum_{j \in J_s} \frac{t_j}{|J_s|} = \sum_{i \in [Q]} \frac{t_i}{Q} = \hat{t} \quad \forall s \in S.
\end{equation}

We argue that a Wardropian Cycle can be constructed for each group separately leading to a fair assignment for this OD pair.
This holds because every driver in a Wardropian Cycle experiences mean travel time equal to mean daily travel time of all routes in the cycle, so will have travel time equal to the mean time of their partition $J_s$, hence equal for every driver on this OD pair. Unfortunately, the following claim is true: 

\begin{claim}
\label{cl:np_hard}
    Dividing a set of routes with time and flow into subsets of equal mean time is $NP$-hard.
\end{claim}
\begin{proof}
    See Appendix \ref{app:NPhard} (reduction from the Subset Sum Problem).
\end{proof}

We note that this approach is applicable only in scenarios involving small flows and travel times, as the associated computational cost causes this solution to be impractical on a scale of an entire city. Nevertheless, multiple heuristics can be applied to the partitioning problem to achieve at least close approximation of such partition in polynomial time.



\subsection{Bounding Time Deviations in Wardropian Cycles}
\label{sub:opt_cycle}

In addition to being short, the cycles shall be fair among participants. We note that 
the 
order of assignments $(A_1, \ldots, A_n)$ matters for the drivers. There is a significant difference between being forced to travel longer for many days with a promise of compensation at the end, and alternating longer and shorter routes day by day to arrive faster on average. In this section, we consider how to reduce inequity throughout the cycle, not only at the end. We introduce the following general optimisation problem. Here, we search for a reordering $\sigma$ of daily assignments $(A_1, \ldots, A_n)$ to $(A_{\sigma(1)}, \ldots, A_{\sigma(n)})$ (within a set of all possible permutations $S_n$ of a daily assignment matrices) such that any cumulative deviation from the mean for any traveller is minimised.

\begin{definition}[Restricted cycle inequity optimisation problem]
\label{def:opt_single}
    Let $(A_1, \ldots, A_n)$ be a sequence of daily assignments resulting in daily travel time deviations $(\bigD_1, \ldots, \bigD_n)$ and $||\cdot||_{\infty}$ denotes the $l_{\infty}$ norm.
    The \emph{restricted cycle inequity optimisation problem} is given by:
    \begin{equation}
    \label{eq:opt_single}
        \min_{\sigma \in S_n} \left\{ \max_{n_s, n_l \in \{1,\dots, n\}} \left|\left| \sum_{j=0}^{n_l-1} \bigD_{\sigma(n_s + j)}  \right|\right|_{\infty} \right\},
    \end{equation}
where $S_N$ is the set of all  cyclical permutations of $\{1,2, \dots, N\}$, i.e. permutations $\sigma$ for which $\sigma(n) := \sigma(n-N)$ for $N < n \le 2N$.
\end{definition}


This problem turns out to be NP-hard, see \ref{Prop_cycleNPhard}.
Therefore, in Proposition \ref{Prop3} below we introduce a heuristic method to approximate the solution to problem \eqref{eq:opt_single} that limits the maximal inequity by the distribution of route times.


\begin{proposition}\label{Prop3}
    Let $\hat{t}$ denote the average single-day travel time of an agent in a Wardropian Cycle with daily assignments $\bigA = (A_1, \ldots, A_n)$ and let $t_{\max}$ and $t_{\min}$  denote the longest and the shortest single-day travel time in $\mathcal{A}$, respectively.
    Let $\hat{t}_{i, l}(\bigA) = \sum_{j \leq l} [A_jT]_i /l$ be the average travel time of agent $i$ after $l$ days of his assignment. 
    One can find a matrix $A'$ and the corresponding initial assignment $\bigA'= (A', PA',\dots,P^{n-2}A',P^{n-1}A')$, where $P$ is the shift matrix, such that
    \begin{equation} \label{eq:opt_all}
        \forall_{i \in Q} \forall_{l \leq n} \left|\hat{t}_{i, l}(\bigA') - \hat{t}\right|\leq \frac{t_{\max}- t_{\min}}{l}.
    \end{equation}
\end{proposition}

We note that although in the general case, the route times might differ significantly, in practice, however (see Fig. \ref{fig:longest_shortest_diff}) these differences are typically small and $t_{max} - t_{min}$ is reasonably bounded. Moreover, since the method presented in this paper is assignment-agnostic, we could use any of the methods of constrained system optimum assignments further decreasing this difference, for example from: \cite{Jahn2005SystemOptimal} and \cite{ANGELELLI2020105016}.

\begin{proof}[Proof of Proposition \ref{Prop3}]
The sketch of the proof goes as follows.
First, we construct a cycle of routes for the first agent that satisfies a strict bound. Then, we assign other agents to the same cycle, yet starting at different positions. We then show that such assignments satisfies bound from eq. \ref{eq:opt_all}. The individual assignments uniquely define the $\mathcal{A}'$.

Let $\mathcal{T}' = (t'_1, t'_2,\dots, t'_n)$ be a sequence of times that agents in a cycle experience during single cycle completion (time $t_k$ of each route $r_k$ is included exactly $q_k$ times). W.l.o.g. we can say: $(t'_1 \leq t'_2 \leq \dots \leq t'_n)$ and subtract $\hat{t}$ from each value, $\mathcal{T} = (t_1, \dots, t_n) := (t'_1 - \hat{t}, \dots, t'_n - \hat{t})$. Note that after subtraction $\sum_{t \in \mathcal{T}} t = 0$. Now we split the series $\mathcal{T}$ into $\mathcal{T}^- = \{t_i \leq 0\, |\, t_i \in \mathcal{T}\}$ and $\mathcal{T}^+ = \{t_i > 0\, |\, t_i \in \mathcal{T}\}$.

Let $\mathcal{O} = (o_1, \dots, o_n)$ be a sequence from sets $\mathcal{T}^+$ and $\mathcal{T}^-$ ordering the routes of the first agent. 
We write $\hat{t}_{1,l} = \frac{\sum_{i=1}^l o_i}{l}$.
We seek such $\mathcal{O}$ that satisfies:

\begin{equation}
\label{eq:order_bound}
    \hat{t} - t_\delta^-/l \leq \hat{t}_{1,l} \leq \hat{t} + t_\delta^+ /l \quad \forall_{l \leq n} 
\end{equation}

where: 
$t_\delta^+ \coloneqq t_{\max} -\hat{t}$ and $t_\delta^- \coloneqq \hat{t} - t_{\min} $.

To create a sequence satisfying eq. \ref{eq:order_bound} we design the following algorithm:
\begin{enumerate}
    \item Take first element  $t_1$ from $\mathcal{T}^+$ and set it as a first element in the sequence $\mathcal{O}$. Let the $\Sigma_\mathcal{O}$ be the sum of all elements currently in $\mathcal{O}$. Define $\mathcal{T}_1^+ = \mathcal{T}^+ \setminus \{t_1\}$ and $ \mathcal{T}_1^- = \mathcal{T}^- $.
    \item To find element $o_i$, pick the first element $t_1$ from $\mathcal{T}_{i-1}^+$ if $\Sigma_\mathcal{O}>0$ and from $\mathcal{T}_{i-1}^-$ if $\Sigma_{\mathcal{O}}<0$.  If $\Sigma_\mathcal{O} = 0$, we can take from either. 
    Place this element at the end of $\mathcal{O}$. Set $\mathcal{T}_{i}^+ = \mathcal{T}_{i-1}^+ \setminus \{t_1\}$ and $\mathcal{T}_{i}^{-} = \mathcal{T}_{i-1}^{-}$ if the element was from $\mathcal{T}_{i}^+$ and reverse otherwise.
    \item Repeat step 2 until both $\mathcal{T}^+$ and $\mathcal{T}^-$ are empty.
\end{enumerate}
To prove the correctness of the algorithm we need to prove, that in any moment we can perform step 2 of the algorithm and it does not exceed the bound from eq. \ref{eq:order_bound}. 

First part is straightforward.
We seek an sequence $\mathcal{O}$ of length $n$, combining the elements from $\mathcal{T}^+$ and $\mathcal{T}^-$. Since the sum of elements in $\mathcal{T}^+ \cup \mathcal{T}^-$ is zero at the beginning of the algorithm, at any moment, if $\Sigma_\mathcal{O}$ is positive set $\mathcal{T}^-$ must be non-empty and if $\Sigma_\mathcal{O}$ is negative set $\mathcal{T}^+$ is non-empty, which guarantees finding element $o_i$ for every $i$.


To see that for every step the bound from eq. \ref{eq:order_bound} is satisfied, we observe that at any moment the sum of $\mathcal{O}$ does not exceed $|\mathcal{O}| \hat{t} + t_\delta^+$ and does not fall below $|\mathcal{O}|\hat{t} - t_\delta^-$. This follows from the fact, that in the algorithm, while changing the sign of the $\Sigma_\mathcal{O}$ we add at most $t_\delta^+$ or subtract at most $t_\delta^-$ what immediately proves the bound.

Next, for agent $i$ we define order $\mathcal{O}_i = (o_i, \ldots, o_n, o_1, \ldots, o_{i-1})$. By eq. \ref{eq:order_bound}, the absolute value of each partial sum of $\mathcal{O}i$ is bounded by $t\delta^+ + t_\delta^-$, which becomes evident upon expanding the definition of the mean :
 \begin{equation*}
     \hat{t}_{i,l-i+1} = \frac{\hat{t}_{1,l}\cdot l - \hat{t}_{1,i-1}\cdot(i-1)}{l-i+1} \leq \frac{\hat{t}\cdot l + t_\delta^+ - \hat{t} \cdot(i-1) + t_\delta^-}{l-i+1} = \hat{t} + \frac{t_{\max} - t_{\min}}{(l-i+1)}
 \end{equation*}

 The same argument being true for the lower bound.

To formally represent sequence $\mathcal{O}$ as an assignment, we create a binary matrix $A$ by setting $1$ for agent $i$ in the column representing any route of time $o_i+\hat{t}$, so that the constraint \ref{eq_A2} is satisfied.
\end{proof}

\subsection{Greedy Assignment Rule minimising Inequity} 
\label{sec:greedy}

The analytical approaches, while satisfactory (they guarantee that cycles exist and are finite), often identify long cycles (with 5\% of cycles in Barcelona longer than 90 days, as we demonstrate in the results). 
To operationalise it and not rely on properties which are not granted (like high GCD of flows), we propose a heuristic \textbf{Assignment Rule} to achieve much faster satisfactory, although not perfect sequences of assignments.
While heuristic might not guarantee full equilibration of mean travel times, examples presented later show that the travel time deviations are quickly reduced. 

\begin{definition}
\label{def:AssignmentRule}
\textbf{Assignment rule} $\Gamma$ is a mapping that generates daily assignments. Let $j+1$ be the considered day, $(A_1, \ldots, A_j)$ past daily assignments ($0$ on the first day) and $\mathcal{D}_j$ accumulated time deviations ($\mathcal{D}_0 = 0$). Given the assignment rule $\Gamma$, the daily assignment $A_{j+1} = \Gamma((A_1, \ldots, A_j), \mathcal{D}_j)$. 
An assignment rule $\Gamma$ is a \emph{Markovian assignment rule} if it is irrespective of $(A_1, \ldots, A_j)$, i.e. for any two past assignments $(A_1, \ldots, A_j)$ and $(B_1, \ldots, B_j)$ which yield the same $\mathcal{D}_j$, $\Gamma((A_1, \ldots, A_j), \mathcal{D}_j) = \Gamma((B_1, \ldots, B_j), \mathcal{D}_j)$.
\end{definition}


By the above definition, the Wardopian Cycle is the result of any assignment rule that achieves $\mathcal{D}_j =0$ on some day $j$. In general, we can consider any rule that aims to reduce the time deviations as a heuristic solution to the multi-day traffic assignment problem. 

To minimise the cumulative discontent in a general setting, we propose the Greedy Assignment Rule and the corresponding algorithm, illustrated in Fig. \ref{fig:methodology}. The idea behind this is simple: drivers who experienced the highest cumulative discontent are assigned the fastest routes tomorrow. This approach is effective, quickly computable, and applicable to many generalised settings (such as varying demand and assignment).

Each day, we sort drivers in decreasing order of their travel time deviation sums (eq. \ref{eq:cum_time_dev}), pairing them to the corresponding route in order. The drivers most favoured so far are assigned to longer routes and vice versa.


We formalise our assignment rule as follows.
Without loss of generality, we assume $t_1 \leq \dots, \leq t_K$. 
Then, we find a permutation $\pi_j$ of the vector $\mathcal{D}_j$ such that $\pi_j (\mathcal{D}_j)$ is decreasing (i.e. each subsequent element of the vector is weakly smaller).
We denote $\pi_{j, l}$ the $l$-th element of this permutation.
We define mapping $T_j: \{1, \ldots, Q \} \ni l \xrightarrow{} k: \sum_{j \leq k} Q_{j-1} < l \leq \sum_{j \leq k} Q_{j}, k \in \{1, \ldots, K \}$, where $Q_j$ are SO flows on routes $1, \ldots, K$.
We construct a permutation matrix $P_j$ such that the ones are in positions $\{(\pi_{j, l}, T_j(l)): l \in \{ 1, \ldots, Q \} \}$.
Then, we define the assignment rule of the Greedy Algorithm as $A_{j+1} = P_j A_j$.

Worth noting is the fact that we do not require the agents set to be the same every day, nor the routes to be equal. The described procedure works on an arbitrary assignment and flows, as we discuss in section \ref{sec:bit}.
Also, the proposed assignment rule is Markovian, as we just consider the $\mathcal{D}_j$.

\begin{proposition} 
\label{prop:greedy}
    The Greedy assignment is the assignment that minimises the inequity for the next day.

    \begin{proof}
        See \ref{app:algo_properties}
    \end{proof}
\end{proposition}

With such an approach, reaching the full equity is not guaranteed (as we will demonstrate in the results that the greedy algorithm tends to oscillate about the equilibrium by a constant); thus, the cycles are not purely Wardropian anymore in a sense of the definition \ref{def:wardropian_cycle}). However, this procedure allows us to rapidly reduce the Inequity of the system while omitting the assumptions of fixed flow and route times. 

Moreover, we provide a simple proof that the greedy algorithm produces bounded cumulative deviations, hence the average time of each agent converges to the average time of the OD pair, with a \emph{very} non-optimal constant. This proof can certainly be significantly improved.

\begin{proposition}
(Average times of agents converge under the Greedy Assignment rule)

For an SO assignment, let $\bigD_0$ denote time deviations, and $M:= ||\bigD_0||_{\infty}$ (the $l_{\infty}$ norm) be the maximal deviation from the average. Let $K^+ := |\{i: [\bigD_0]_i \ge 0\}|$ be the number of non-negative deviations and let $K^-:= |\{i: [\bigD_0]_i < 0\}|$ be the number of negative deviations, where $K^+ + K^- = Q$. Then for every agent i and every day up to $J$ we have 
\begin{equation}
\label{eq_Dbounded}
-M(K^- +1)\le \sum_{j=1}^J [\bigD_J]_i \le M(K^+ +1).
\end{equation}
\end{proposition}
\begin{proof}
Suppose the contrary. Let $J_0$ be the first day on which \eqref{eq_Dbounded} is violated. Suppose first that for some agent $i_0$ we have $\sum_{j=1}^J [\bigD_j]_{i_0} > M(K^+ +1)$. Then on day $J_0-1$ we have $\sum_{j=1}^{J_0-1} [\bigD_j]_i > MK^+$ for at least $K^- +1$ agents since agent $i_0$ was at least $K^- + 1$-st in the order of decreasing cumulative deviations to receive a positive deviation on day J.  
On the other hand, for every among the remaining $K^+ - 1$ agents we have  $\sum_{j=1}^{J_0 - 1} [\bigD_j]_i > -M(K^- + 1)$.
Consequently, 
\begin{equation*}
0 = \sum_i \sum_{j=1}^{J_0 - 1} [\bigD_j]_i > (K^- + 1)MK^+ + (K^+ - 1) (-M) (K^- + 1) > M(K^- + 1),
\end{equation*}
which is a contradiction. The case where on day $J_0$ we have $\sum_{j=1}^J [\bigD_j]_{i_0} < -M(K^- + 1)$ is similar. 
\end{proof}

This bound on each agent deviation works regardless of the assignment on the underlying OD pair. In Results (Figure \ref{fig:conv_greedy}), we show that for OD pairs from real data, the convergence is orders of magnitude faster.



\subsection{Reducing the Discontent Progression} 

\label{sec:compliance}

As we argued, the cycle $\mathcal{A}$ shall nullify the $\bigD_j$ (eq. \ref{eq:time_deviations}) and $\mathcal{D}_j$ (eq. \ref{eq:cum_time_dev}). 
Proposition \ref{Prop3} improved the cycle by bounding the deviations from the mean for every driver. In this section, we dive deeper into this topic and 
advocate for minimising also the $\mathcal{H}_J$ measure (eq. \ref{eq:history}) while constructing Wardropian Cycles.

Minimising cumulative deviations or/and cumulative discontent, as per Definition \ref{def:cum_discontent}, amounts to balancing the travel times above average in terms of cumulative deviations or/and cumulative discontent by travel times below average. However, even for the optimal permutation, there is always some residual cumulative deviation/discontent, different across different users. Indeed, example \ref{ex:discontent} shows that after a single iteration of a cycle, certain drivers may come out more discontent than others.







Unfortunately, minimising the discontent progression $\mathcal{H}$ turns out to be NP-hard, yet we propose heuristics to reduce it.
We may track $\mathcal{H}$ in the cycle and reduce it across the cycles. Namely, by generalising the idea of greedy assignment (Sec. \ref{sec:greedy}), i.e. the drivers with the worst history shall be assigned in the next cycles such that their $\mathcal{H}$ is minimised.

This can be achieved via the Permuted cyclical assignment (def. \ref{def:permuted}) and the inter-cycle assignment of drivers to \emph{positions in cycle} with a problem which minimises it (def. \ref{def:intercycle}). A permuted cyclical assignment repeats the cyclical assignment; however, it changes the positions of the drivers with every iteration of the cycle.

\begin{definition}[Permuted cyclical assignment]
\label{def:permuted}
    A \emph{permuted cyclical assignment} is, for a given $\zeta$ and sequence $(A_1, \dots, A_{\zeta})$, an assignment 
    \begin{eqnarray*}
    &&(\Sigma_1 A_1, \Sigma_1 A_2, \dots, \Sigma_1 A_\zeta, \\ &&\Sigma_2 A_1, \Sigma_2 A_2, \dots, \Sigma_2 A_\zeta, \\&&\Sigma_3 A_1, \Sigma_3 A_2, \dots, \Sigma_3 A_\zeta, \dots)
    \end{eqnarray*}
    where $\Sigma_1, \Sigma_2, \dots$ is some sequence of matrices permuting the drivers.

\end{definition}

\begin{definition} [Intercycle optimisation problem]
\label{def:intercycle}
Given a Wardropian Cycle and a given previous assignment to \emph{positions in the cycle} in the first iteration of the cycle, determine the assignment of drivers to positions in the cycle in the second and further iterations of the cycle.
Formally:
\begin{equation}
    \min_{\sigma \in S_n} \left\{ \sum_{i \leq Q}  \left[ \max_{j\leq J} -[\mathcal{H}_J]_{i,j} \right] \right\}
\end{equation}

\end{definition}
Without going into details, we remark that an efficient heuristic to this problem could be a generalisation of the Greedy Assignment Rule described in Section \ref{sec:greedy}. For instance, for the driver $1$ in permutation $\sigma_1$ in Example \ref{ex:discontent}, who came out with (both final and maximum) discontent $4$, it may be preferred to be assigned to position $2$ in the cycle in the second iteration of the cycle to compensate for the unpleasant experience in the first cycle.

Furthermore, one could also ask what the optimal Wardropian Cycle is and whether it has to be derived from a single permutation. Indeed, it turns out that admitting general cyclical assignments of the same length one could decrease the worst cumulative deviations for the price of complexity of the assignment pattern, see Example \ref{ex:multi_better}.


\subsection{Discretisation}
\label{sec:discrete}

While the solution to the Assignment (eq. \ref{eq:Assignment}) assigns real-valued flow to real-valued path flows $Q_k$, our method is inherently discrete and operates on individual drivers. Therefore, to operationalise it, we need an assignment solution in discrete space:

     \begin{equation}
     [Q_k \in \mathbb{N}, t_k]_{k \leq K} = \text{Assignment}(Q \in \mathbb{N}) 
     \label{eq:AssignmentDiscrete}
 \end{equation}

 We note that the standard Frank-Wolfe solutions \cite{TrafficAssignment_FrankWolfe2021}, which we use, solve the assignment in continuous space. Yet, rounding flows $Q_k$ to the nearest integer yielded a sufficiently accurate approximation.

\subsection{Cyclical User Equilibrium}
\label{sec:new_wardrop}

    
    


\begin{definition}
    \label{def_cyclical}
    Drivers on a given OD pair are in \textit{Cyclical User Equilibrium}, if for some period $\zeta$ they follow a sequence of assignments $\mathcal{A}^{WC} =(A_1, \ldots, A_{\zeta})$ such that every driver gains from switching from $\mathcal{A}^{UE}$ to $\mathcal{A}^{WC}$.

    \begin{subequations} \label{eq:CUE}
\begin{align}
    \sum_{j=1}^{\zeta} \mathbf{D}_j &= \mathbf{0}, \label{eq:CUE1} \\
    \forall_{i \leq Q} \left[ \sum_{j=1}^{\zeta} A_j \mathbf{T} \right]_i &< t^{UE} \cdot \zeta \label{eq:CUE2}
\end{align}
\end{subequations}

    where $\bigT=[t_1, \ldots, t_K]^T$ and $t^{UE}$ is the  User Equilibrium travel time on this OD pair.
\end{definition}

This happens when for given OD pair SO assignment yields shorter travel time better than UE and cyclic rule $\mathcal{A}$ forms a Wardropian Cycle. 

Below we discuss the basic theory of Cyclical assignments and equilibria as defined above in terms of standard economical notion of Pareto optimality.
\begin{definition}
Let $\mathcal{A} = (A_1,\dots, A_\zeta)$ and $\mathcal{\tilde{A}} = (\tilde{A}_1,\dots, \tilde{A}_{\tilde{\zeta}})$ be two cyclical assignments. We say that $\mathcal{A} \le \mathcal{\tilde{A}}$ if for every driver $i$ we have
\begin{equation*}
    \left[\sum_{j=1}^\zeta A_j \mathbf{T} /\zeta \right]_i \le \left[\sum_{j=1}^{\tilde{\zeta}} \tilde{A}_j \mathbf{T} /\tilde{\zeta} \right]_i
\end{equation*}
We say that $\mathcal{A} < \mathcal{\tilde{A}}$ if $\mathcal{A} \le \mathcal{\tilde{A}}$ and there exists $i$ such that $\left[\sum_{j=1}^\zeta A_j \mathbf{T}\right]_i < \left[\sum_{j=1}^{\tilde{\zeta}} \tilde{A}_j \mathbf{T}\right]_i$.
\end{definition}

\begin{definition}
    A cyclical assignment $\mathcal{A}$ is \emph{Pareto-optimal} if it is a minimal element of the relation $\le$, i.e. there exists no cyclical assignment $\mathcal{\tilde{A}}$ such that $\mathcal{\tilde{A}} < \mathcal{{A}}$.  
\end{definition}

\begin{example}
\begin{enumerate}
\item[i)] A User Equilibrium with every driver choosing the same route every day is a cyclical assignment with period $\zeta=1$.
\item[ii)] A User Equilibrium with drivers switching routes day-to-day is a permuted cyclical assignment (def. \ref{def:permuted}).
\item[iii)] Let for a given fixed OD pair $\mathcal{{A}}^{\text{UE}}$ correspond to a (possibly non-unique) User Equilibrium (with period $1$),  $\mathcal{{A}}^{\text{SO}}$ correspond to system optimum (with period $1$) and $\mathcal{{A}}^{\text{WC}}$ correspond to a Wardropian Cycle with some period ${\zeta}$ derived from the given System Optimum. Then 
\begin{itemize}
    \item $\mathcal{A}^{\text{UE}}$ and $\mathcal{A}^{\text{SO}}$ are usually incomparable, i.e. neither $\mathcal{A}^{\text{SO}} \le \mathcal{A}^{\text{UE}}$ nor $\mathcal{A}^{\text{UE}} \le \mathcal{A}^{\text{SO}}$, i.e. some drivers travel faster in \emph{UE} and others in \emph{SO}.

    \item $\mathcal{A}^{\text{WC}} \le \mathcal{A}^{\text{UE}}$,
    \item $\mathcal{A}^{\text{WC}}$ and $\mathcal{A}^{\text{SO}}$ are incomparable,
\end{itemize}
 
\item [iv)] $\mathcal{A}^{\text{UE}}$ is usually not Pareto-optimal (unless it coincides with system optimum), while both $\mathcal{A}^{\text{SO}}$ and $\mathcal{A}^{\text{WC}}$ are Pareto-optimal.  
\end{enumerate}
\end{example}

\begin{remark}
Note that if we allow cyclical assignments only of length one, $\mathcal{A^{UE}}$ is Pareto optimal as long as it is the best User Equilibrium. In particular, this is the case when the User Equilibrium (Nash Equilibrium) is unique, see e.g. \cite{dafermos1971extended, smith1979existence}. In view of this, the essence of Definition \ref{def_cyclical} is that $\mathcal{A}^{\text{UE}}$ is usually no longer Pareto optimal when cycles \emph {longer} than $1$ are allowed. 
\end{remark}

\begin{proposition}
A Cyclical User Equilibrium is Pareto-optimal within the class of Cyclical User Equilibria if and only if it is a Wardropian Cycle derived from the System Optimum for a given pair OD. 
\end{proposition}
\begin{remark}
    There exist Cyclical User Equilibria which are not Pareto-optimal, derived from assignments which are somewhere in the price of anarchy gap, i.e. between SO and UE in terms of average travel time. 
\end{remark}

The proposed Cyclical User Equilibrium concept sheds a new light on the classic Wardopian principles. Allowing for a longer cycle $\zeta>1$ and comparing average travel times rather than single realisations leads to a new perspective. The comparison of two assignments $\mathcal{A}$ allow identifying the \emph{Pareto-optimal} one. If we depart from SO assignment and our cycle is Wardopian users travel both optimal and with the same travel time. This allows us to paraphrase Wardrop and say that in CUE: \emph{No user has incentive to switch position in cyclical assignment with other user on the same OD pair (eq.\ref{eq:CUE1}) and no user has incentive to opt out of the cycle and revert to User Equilibrium (eq.\ref{eq:CUE2}).}

\subsection{OD-fair System Optimal Assignment}
\label{sec:CSO}

To make the \emph{Wardopian Cycles} system-wide applicable, we need to make sure that the system will be in Cyclical User Equilibrium (def. \ref{def_cyclical}); in particular, average travel times for each OD pair shall be shorter than in User Equilibrium (eq. \ref{eq:CUE2}).
However, setting a city-wide system flow pattern to reach the global optimum (minimising the total travel time) often leads to a worse performance on some OD pairs than in User Equilibrium. In our  Barcelona case, $17\%$ of OD pairs are worse off in SO assignment despite overall system profit. Classical formulations of SO as an optimisation problem while minimising global objective do not control whether on some OD-pairs travel time increases. To this end, we propose an \textbf{OD-fair System Optimal assignment} by adding constraint to the classical Equilibrium formulation, introduced by \cite{BeckmannTAP}, that checks if: \textit{All OD pairs in the assignment have the average travel at most equal to one in the User Equilibrium assignment}:
\begin{equation}
   \sum_{k \in K} Q_k t_k / Q \leq t_{UE}, \forall(o,d)
\end{equation}
Unfortunately, this cannot be easily implemented, due to the non-convexity of this constraint in a multivariable domain. Here, we simply formulate the problem, leaving the practical implementation to future studies.





Alternatively, the pairs benefiting from the scheme may transfer produced benefits to the drivers travelling disadvantaged pairs. This could be done via monetary or non-monetary transfers, with little chances of success (as some systematically disadvantaged users are unlikely to be happy with monetary compensations \cite{bohren2022systemic}).

\subsection{Inequity Minimising Daily Assignment Rule}
\label{sec:bit}

While the cycles are neat, formal, and reduce inequity in finite time, they become inadequate as soon as the setting becomes more realistic. In particular, like typically in transport systems, when travel times $t_k$ vary from day to day, or flows $Q$ and the assignment solution $Q_k$ change in time. Any of such natural fluctuations make the proposed Wardropian Cycles theory inapplicable.

To this end, we propose the concept of \emph{Inequity minimising Daily Assignment Rule} applicable in more general contexts. Now, we generalise the \emph{Assignment Rule}  (def. \ref{def:AssignmentRule}) for the so-called \emph{bit} of the traffic flow.
The \emph{bit} $b$ is a flow of $Q^b$ vehicles between the two arbitrary points in the network $O_b$ , $D_b$ for which an externally computed \emph{Assignment} is given, distributing the bit flow to paths $Q_k$. If we know the history of (normalised) inequity of each driver in this bit $\bar{I}_{i \in Q^b, j}$ until today $j$, the greedy assignment rule (prop. \ref{prop:greedy}) can be applied to reduce this inequity.


With such formulation, any Markovian rule $\Gamma$, in particular the Greedy assignment (Sec. \ref{sec:greedy}), may be applied to any selected part of the city, with any flow, also when only part of the flow is controlled. The drivers' composition $Q_b$ does not need to be identically repeated every day and travel costs may vary. The flow may be composed of regular drivers with a long history of participation along with occasional drivers travelling rarely. If only one can predict the total flow $Q_b$ and compute the desired assignment to the paths $Q_k$, we may assign the drivers participating in that round of the scheme so that their equity will be minimised in long-term. Based on Proposition \ref{prop:greedy}, we may use greedy assignment and allocate drivers with the greatest inequity so far to the shortest paths. This requires a normalised form of inequality $\bar{I}$ (eq. \ref{eq:intequitynormalized}), which tracks the cumulated history of travelling of a driver along a path of various lengths with various travel times:

\begin{equation}
    A_{j+1} = \Gamma((A_1, \ldots, A_j), \mathcal{D}_j). 
\end{equation}



\subsection{Method summary and limitations}
To conclude the method, we note that:
\begin{enumerate}
    \item The notion of SO in equation \ref{SO} can be replaced with any assignment, i.e. distribution of the total OD flow to paths. Here, we focus on System Optimal assignment as most illustrative and appropriate, yet in general any assignment can be equalised with the proposed framework. While many assignments will not be attractive, as they will yield travel times longer than in user equilibrium, presumably some others (e.g. minimising total mileage or CO2 emissions) can be both desired for policymakers and attractive for users.
    \item We require the strict mapping between flows and travel times $t_k=f(q_k)$. Naturally, we applied the classic convex BPR functions, for which we could compute city-wide assignment in Barcelona and other cities. However, the core method works with any model where travel times depend on total flow (which implies either macroscopic models or low sensitivity to route swapping in microscopic models). 
    \item Although we operate on the continuous fluid-like macroscopic flows, we propose the method for individual agents (drivers) as discussed in Section \ref{sec:discrete}.
    \item We reduce the travel costs to travel times, which (without the loss of generality) can be replaced with more general cost formulations. The problem may arise when individual path costs are differently perceived by the drivers and the order of path attractiveness varies between individuals.
    \item We define the Cycles for an OD pair flows, but we generalise it for a \emph{bit} of flow (Section \ref{sec:bit}) - which is more general and does not require full compliance nor coverage.
\end{enumerate}






\section{Results}
We start with a toy network example to showcase the Wardropian Cycle, as in Figure \ref{fig:methodology}, where we assign $18$ vehicles to three paths and analyse the cycle length and inequity when drivers follow distinct assignments.
Next, we apply our concept to real-world networks. 
We introduce our approach to the city of Barcelona, starting with the OD-level analysis, follow with system-wide analysis. 
Finally, to see broader potential impact, we report results from four cities from the Transportation Networks repository. 




\subsection{Illustration}
To illustrate how Wardropian Cycles work, we introduce the situation, as in Figure \ref{fig:methodology}. 
We analyse an abstract OD pair with $Q=18$ vehicles optimally assigned to three routes of various travel (see Table \ref{tab:toy_example}). 
We first show a full Wardropian Cycle, and then we reduce its length and inequity by adopting methods from Sections \ref{sec:optim_flow}, \ref{sec:optim_times}, and \ref{sec:greedy}.

\begin{table}[ht!]
\centering
\begin{tabular}{lcc}
\toprule
\textbf{Route} & \textbf{Flow} & \textbf{Time} \\
\midrule
Route 1. & 4 agents & 15 min \\
Route 2. & 6 agents & 14 min \\
Route 3. & 8 agents & 9 min \\
\bottomrule
\end{tabular}
\caption{Example system-optimal assignment for OD pair with $18$ drivers and average travel time $\hat{t} = 12 \text{min}$.}
\label{tab:toy_example}
\end{table}

\paragraph{Cycle length}
Full Wardropian Cycle (Sec. \ref{sec:upper_boundary}) comprise the sequence of 18 shift matrices, i.e. a cycle lasting $18  = 4+6+8 = Q$ days. 
We can improve this result by finding the $M = \text{GCD}(4,6,8) = 2$ and permuting agents in the cycle in batches of $2$ daily, achieving the cycle length $9$, a twofold improvement compared to the initial length. 

Next, we can observe that if an agent travels one day on Route $1$ and then on Route $3$, his average travel time is equal to $(15+9) /2 = 24/2 = 12\text{ min}$, that is, equal to the average time of this OD pair. This is the condition for the partition method. This means that we can separate the flow from Route 1, coupling it with equal flow from Route 3, achieving a separate cycle of length 2 for 8 agents in this OD pair. The remaining 10 agents can be split again into separate Wardropian cycles that with GCD reduction last for 5 days. Eventually, instead of a single 9-day cycle, we have two: one of length 2 for 8 agents, and the other of length 5 for 10 agents. To further improve it, one can apply inter-cycle permutations from Sec. \ref{sec:compliance} and swap drivers between groups every-cycle to reduce $\mathcal{H}$.

\paragraph{Inequity of the cycle}
Apart from the cycle length, we also analyse the Inequity experienced by drivers. 
We compare full cycles defined by the shift matrix (Sec. \ref{sec:upper_boundary}), randomly ordered permutation of these cycles, and systematic reduction from Proposition \ref{Prop3}. 
We present results in Figure \ref{fig:prep3_res}. 
The achieved bound of time deviations allows us to reduce the impact of daily assignments on the cumulated time deviations (eq. \ref{eq:cum_time_dev}) for individual drivers.

\begin{figure}[ht!]
    \centering
    \includegraphics[width=0.7\linewidth]{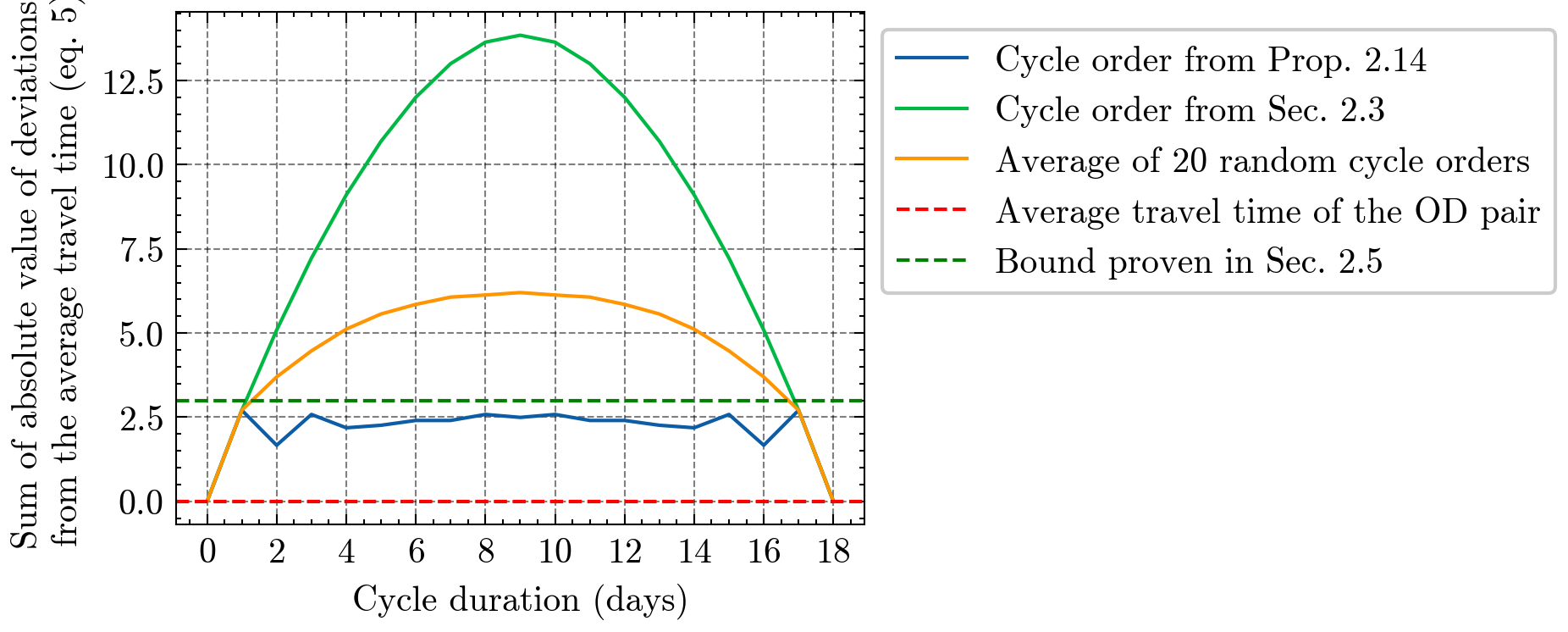}
    \caption{Sum of absolute deviations from the average travel time for the OD pair from Fig. \ref{fig:methodology} during the 18-day cycle after which the travel times equalize. However, the total deviations strongly depend on the ordering. The maximal value ranges from 12.5 when naive ordering is used, as in \ref{sec:upper_boundary} (green), the random (yellow) reduces it around 6. Notably, the bound proven in Sec. \ref{sub:opt_cycle} (dashed line) is conserved with the method proposed (blue line).
}
    \label{fig:prep3_res}
\end{figure}

\subsection{City-scale application}
We tested the concept of Wardopian Cycles for 2078 OD pairs in a Barcelona use-case, with the network topology and demand obtained from the standard repository \cite{TransportationNetworks}. 
To compute static UE and SO assignments we used the Python framework by \cite{TrafficAssignment_FrankWolfe2021}.
Our generic and reusable Python functions implementing the idea of the cycles are available at a public repository \cite{Wardropian_cycles}). We use them to compute analytical and algorithmic cycles for each of OD pairs.

As postulated, the finite Wardropian Cycles exists in Barcelona.
We analyse their length, convergence, and resulting inequity.
Expectedly, the first analytical approach, while guaranteeing perfect convergence, produces long cycles (see Table. \ref{tab:stats}).
The mean cycle is over 31 days long and the median length is 16 days. 95\% of the cycles are shorter than 100 days.
The GCD-based method reduced cycle lengths. The median decreased to 11 days, nonetheless the longest cycle remained unacceptably long (777 days before reaching equity).
Conversely, the Greedy Assignment Rule heuristic quickly reaches satisfying levels of the travel time imbalance; yet, for some OD pairs, fails to achieve a perfect equity.


\begin{figure}[ht!]
    \centering
    \subfloat[\label{fig:UE_to_SO}]{
        \includegraphics[width=0.45\linewidth]{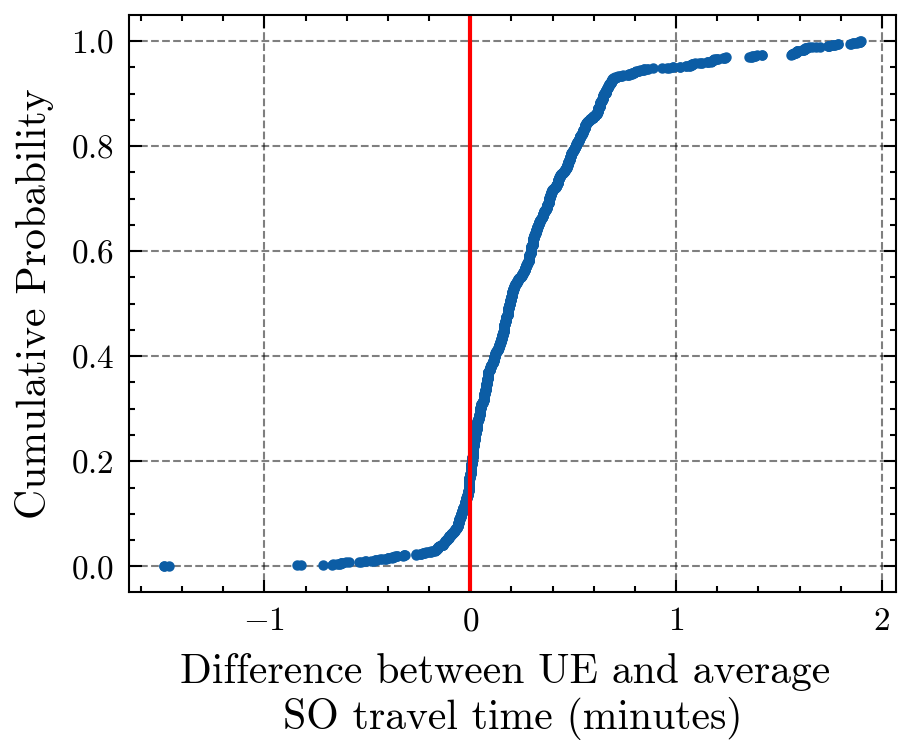}
    }
    \hfill
    \subfloat[\label{fig:longest_shortest_diff}]{
        \includegraphics[width=0.452\linewidth]{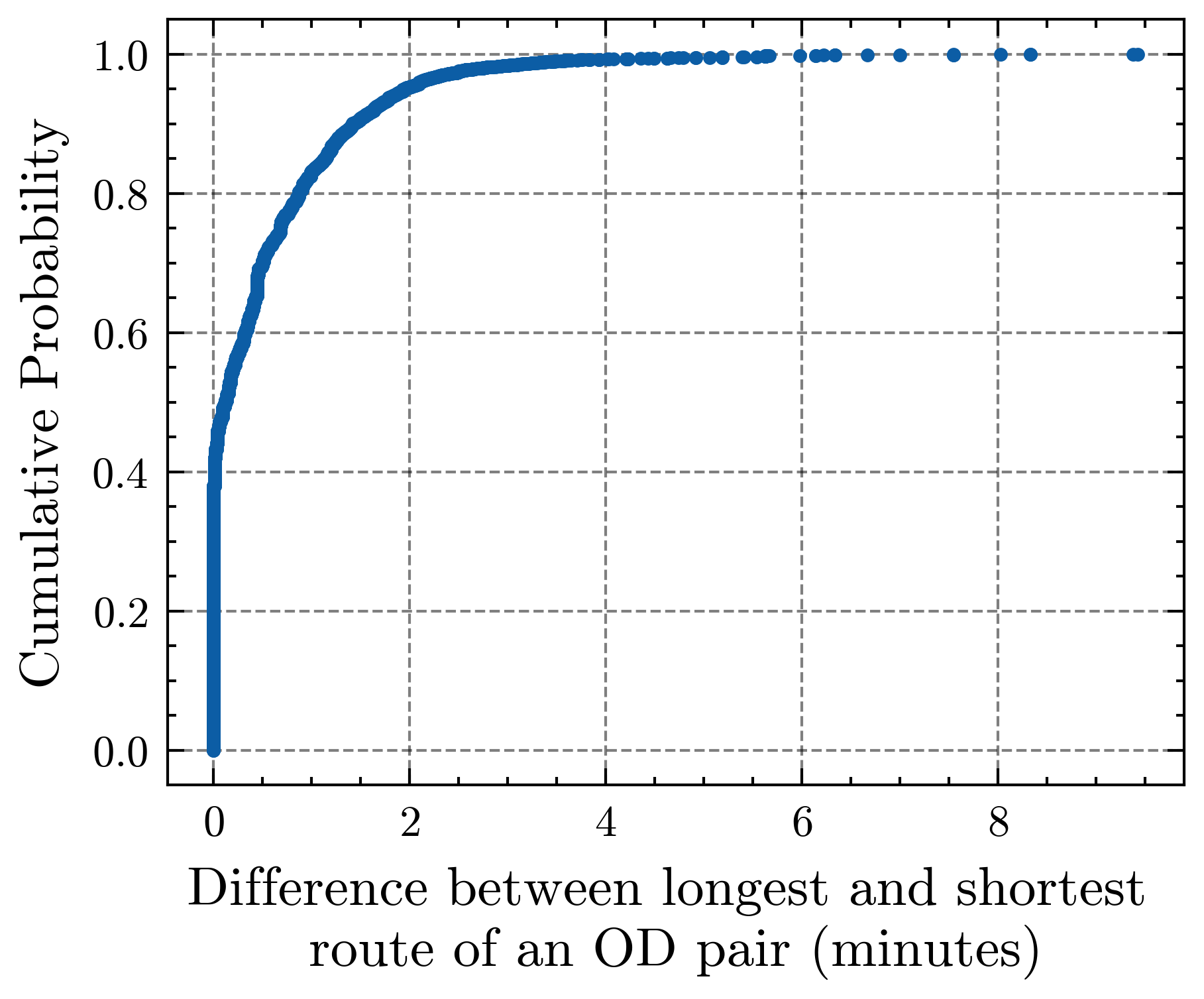}
    }
    \caption{System-wide results for Barcelona: difference between UE and SO travel times per OD (left) and difference between shortest and longest path per OD (right). For most OD pairs, system optimal assignment travel time is shorter than EU, but in Barcelona, nonetheless 17.5\% of OD pairs are worse off if SO assignment is applied (a). Difference between shortest and longest path in SO is below 1 minute for 80\% of OD pairs, and it exceeds 6 minutes only for handful of OD pairs (b).}
\end{figure}

We show the system-wide impact of moving from UE to SO. Applying the assignment that minimises the total travel time improves the system globally by 3\%, which translates to 670 hours of travel time ever day. Fig. \ref{fig:UE_to_SO} shows how the benefits are distributed across the OD pairs. One can notice that not all OD pairs benefit from SO.
Yet, for 82.5\% of the OD pairs, users benefit from complying to our cyclical assignment and on average arrive faster. Others are not in Cyclical User Equilibrium, as they do not meet the criteria of travelling faster than in UE (eq. \ref{eq:CUE2}). To eliminate this issue, we postulate the OD-fair System Optimal Assignment (Sec. \ref{sec:CSO}). 
In Fig. \ref{fig:longest_shortest_diff}, we report differences between the shortest and the fastest path in OD. For 80\% of OD pairs it is shorter than one minute, and it reaches maximum slightly below 10 minutes.
In Fig. \ref{fig:GCD_opt} we report the distribution of cycle lengths computed with analytical methods for 2078 OD pairs in Barcelona. As we see, the cycles exist; yet, they are often too long to be acceptable for drivers.


\begin{table}[ht!]
\centering
\sisetup{table-format=3.2} 
\begin{tabular}{l
                S[table-format=3.0]
                S[table-format=3.2]
                S[table-format=3.0]
                S[table-format=3.2]
                S[table-format=3.0]
                S[table-format=3.1]}
\toprule
\textbf{Metric [days]} & \textbf{Max} & \textbf{Mean} & \textbf{Median} & \textbf{Std Dev} & \textbf{75th Pctl} & \textbf{95th Pctl} \\
\midrule
Full cycles (Sec.~\ref{sec:upper_boundary}) & 777 & 31.37 & 16 & 52.06 & 34 & 100.3 \\
GCD cycles (Sec.~\ref{sec:GCD})             & 777 & 25.94 & 11 & 47.81 & 28 & 91 \\
\bottomrule
\end{tabular}

\caption{Cycle length statistics for Full cycles and GCD Cycles Methods. GCD optimisation improves results, but several cycles are still too long to accept.}
\label{tab:stats}
\end{table}


\begin{figure}[ht!]
    \centering
    \subfloat[\label{fig:GCD_opt}]
    {
        \includegraphics[width=0.46\linewidth]{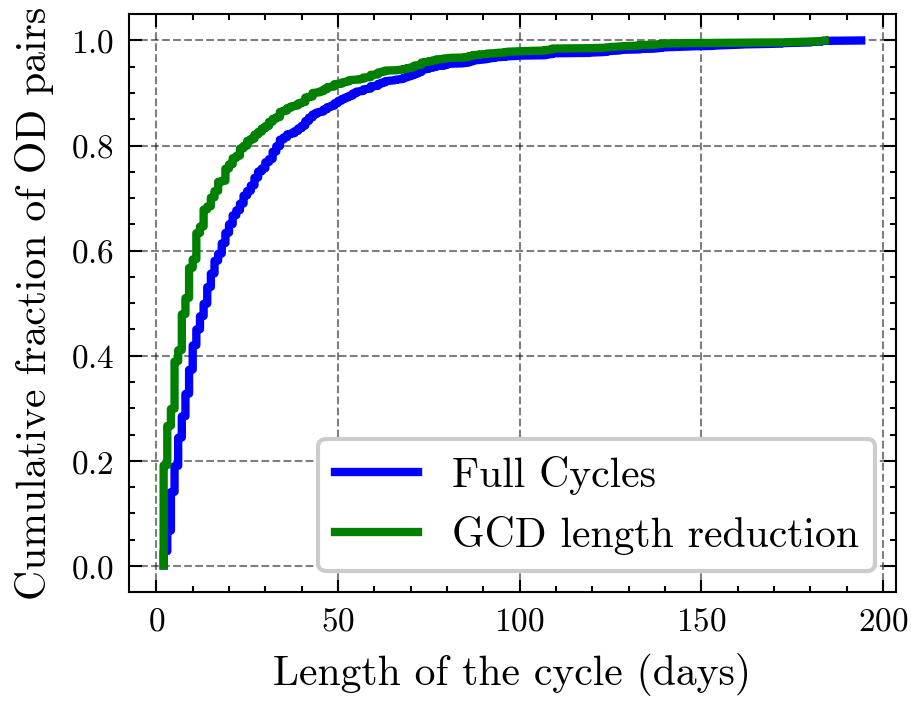}
    }
    \hfill
    \subfloat[\label{fig:greedy_diff}]
    {
        \includegraphics[width=0.46\linewidth]{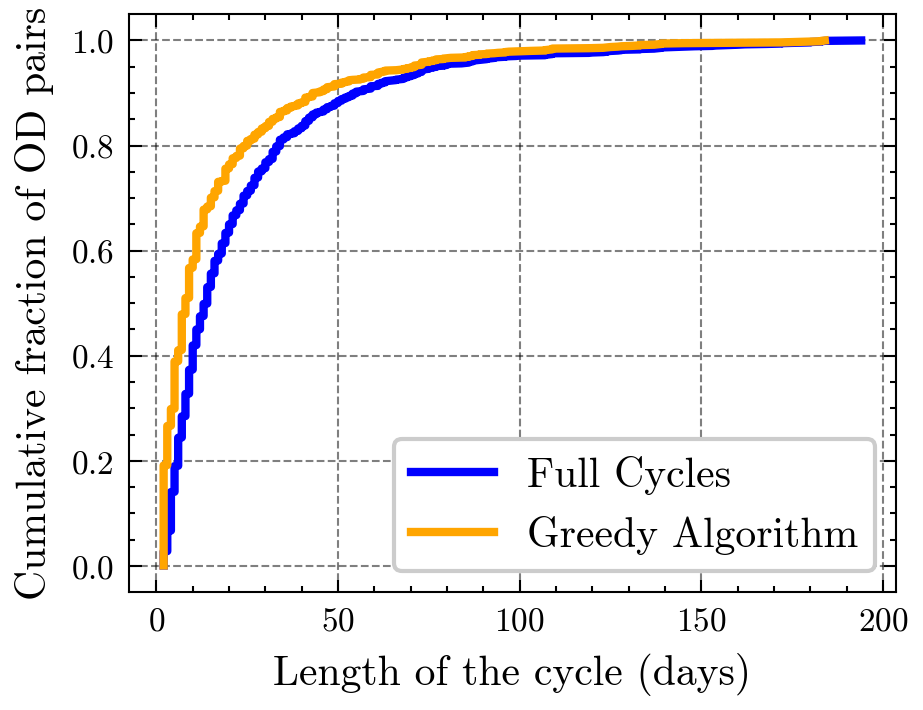}
    }
    \caption{Comparing distributions of cycle lengths in Barcelona with different methods. The full cycles, applying the shift matrix from Sec. \ref{sec:upper_boundary} (blue) can be shortened both with GCD reduction (Sec. \ref{sec:GCD}), marked green on a) and Greedy assignment (Sec. \ref{sec:greedy}), marked yellow on b). 
    Both plots depict only those pairs, that converged in a Greedy Algorithm and, for visibility, we omitted 35 out of 2078 OD pairs that have analytical convergence longer than 200 days.}
\end{figure}

For most OD pairs, the Greedy Assignment Rule improves the analytical solutions as presented in Fig. \ref{fig:greedy_diff}.
Although perfect equity may not be fully achieved, our algorithm significantly reduces inequalities within just a few days. 
In five days, the total inequity (eq. \ref{eq_inequity}) is more than halved and stabilises soon afterwards (Figure \ref{fig:conv_greedy}). As a performance metric, we use the distribution of the maximal and the minimal experienced average time for each OD, since it represents how uneven the agents on this OD pair are. As can be seen in Table \ref{tab:inequity_stats} and Figure \ref{fig:conv_greedy}, we can offer a reduction in the inequity by as much as 88\% in 20 days with extreme values approaching the mean.

\begin{figure}[ht!]
    \centering
    \includegraphics[width=0.5\linewidth]{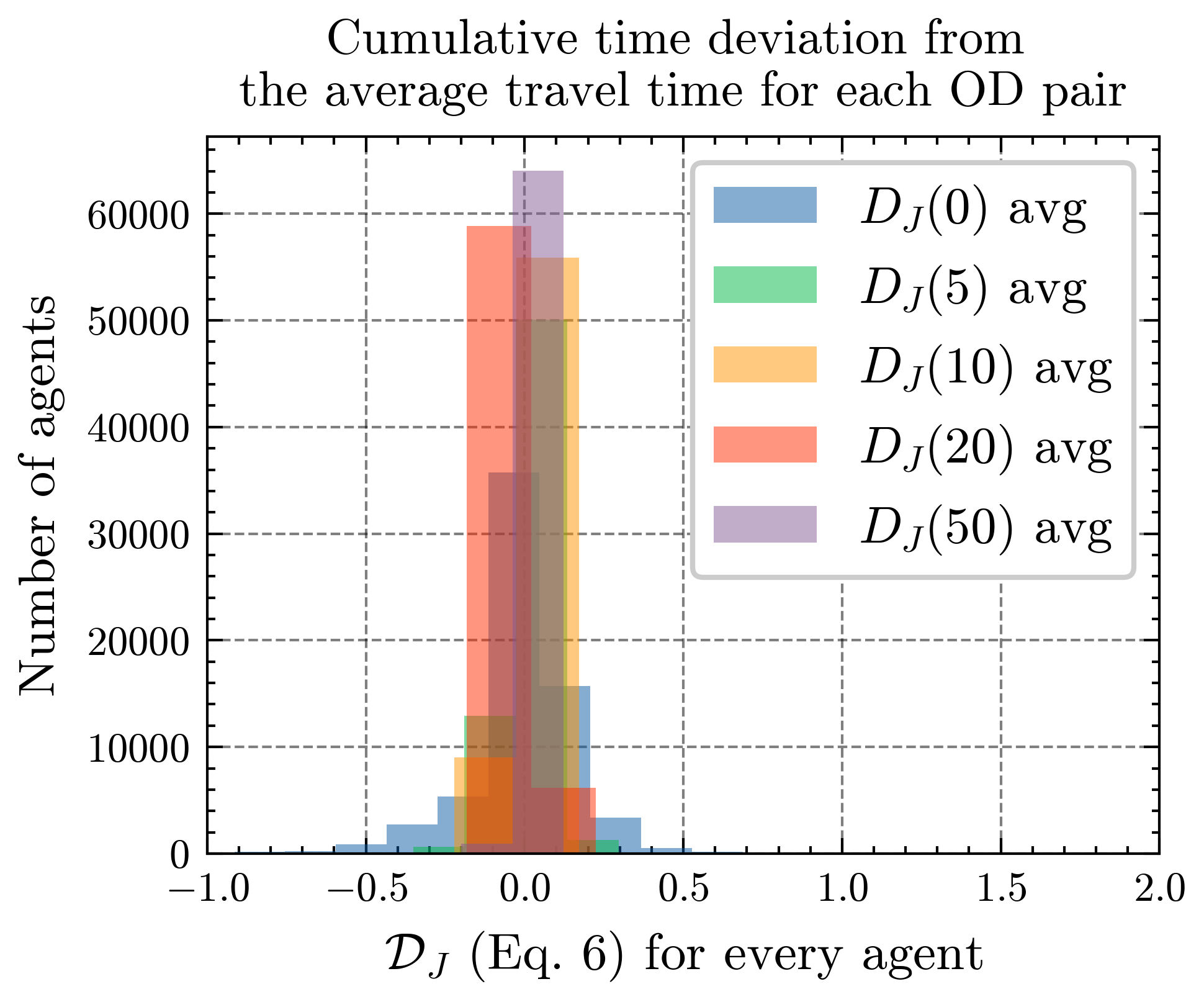}
    \caption{Normalized inequity distribution among agents in the system and its evolution in the cycle. While initially the cumulated time deviations $\mathcal{D_0}$ are unequally distributed, they quickly narrow and differences between faster and slower travellers are diminished.}
    \label{fig:D_J_avg}
\end{figure}

\begin{figure}[ht!]
    \centering
    \includegraphics[width=0.6\linewidth]{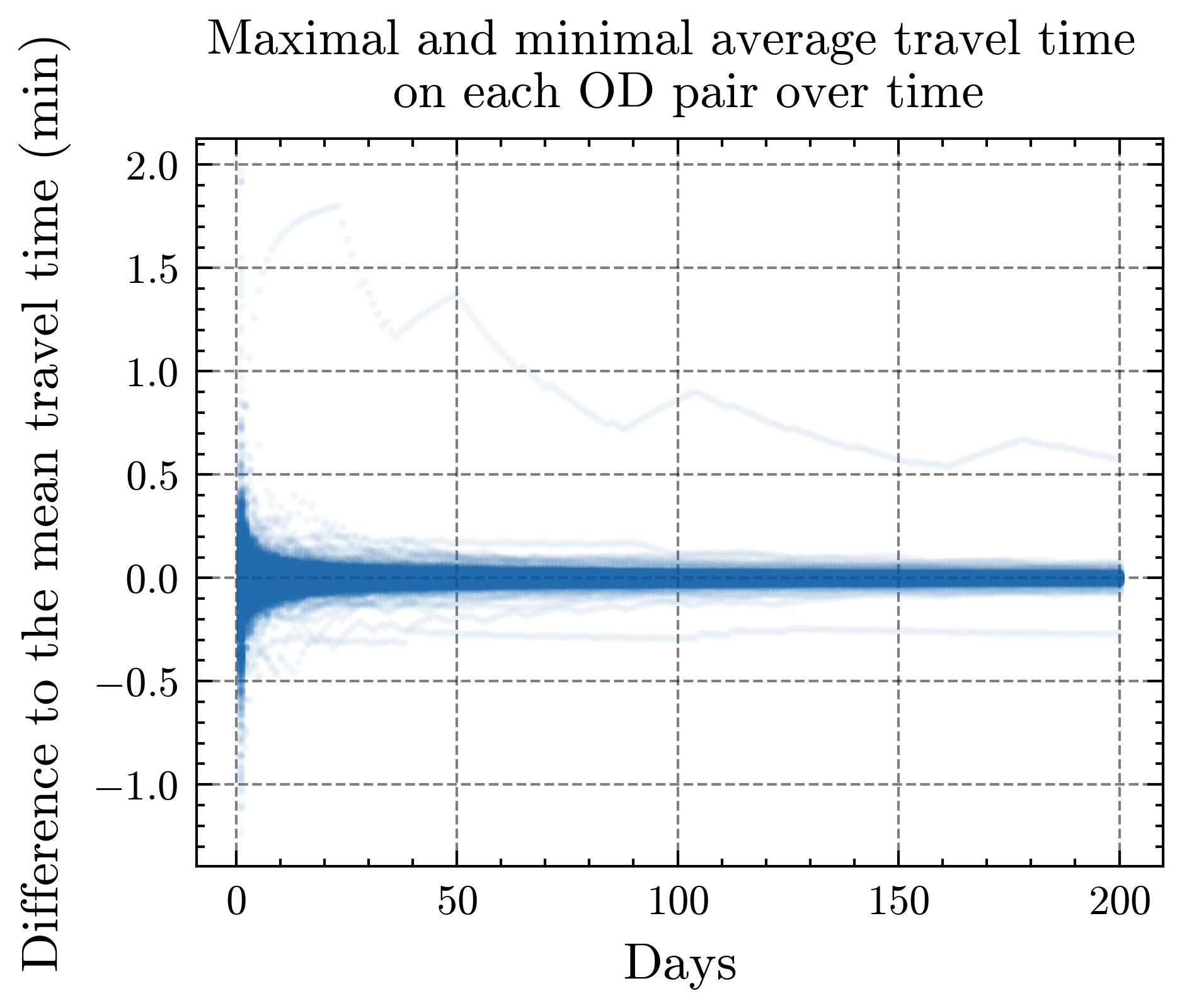}
    \caption{Applying the greedy algorithm steadily reduces the gap between the outliers in the system optimal assignment. For every day and OD pair, the plot shows the minimal and maximal average travel time over all agents on a given pair and day. The cumulated differences in average travel time (Y-axis) cancel out over time (X-axis). While overall inequity stabilises after several days, OD pairs typically reach a day with inequity as low as one-tenth of the initial value in just two or three weeks of simulation.}
    \label{fig:conv_greedy}
\end{figure}

\begin{table}[ht!]
\centering
\begin{tabular}{
    >{\raggedright\arraybackslash}p{2.8cm}
    S[table-format=1.4]
    S[table-format=1.4]
    S[table-format=1.5]
    S[table-format=1.5]
    S[table-format=1.5]
}
\toprule
\textbf{Metric} & $\overline{I_1}$ & $\overline{I_5}$ & $\overline{I_{10}}$ & $\overline{I_{20}}$ & $\overline{I_{50}}$ \\
\midrule
Max          & 8.003  & 5.2910  & 3.00600  & 6.11300  & 4.99000 \\
Mean         & 0.1085 & 0.0175  & 0.00690  & 0.00530  & 0.00310 \\
Median       & 0.0143 & 0.0014  & 0.00036  & 0.00010  & 0.00002 \\
Std Dev      & 0.4517 & 0.1384  & 0.07170  & 0.13480  & 0.10960 \\
75th \%ile   & 0.0648 & 0.0072  & 0.00200  & 0.00070  & 0.00014 \\
95th \%ile   & 0.4179 & 0.0475  & 0.01810  & 0.00700  & 0.00190 \\
\bottomrule
\end{tabular}

\caption{Distribution of an Inequity (eq. 8) over all OD pairs in Barcelona. Initially and after 5, 10, 20 and 50 days of applying Greedy Assignment Rule. Several outliers heavily influence the distribution, what is visible in the \textbf{Max} and \textbf{Std Dev}. Those are the pairs that involve many drivers and routes, slowing down the convergence of the algorithm.}
\label{tab:inequity_stats}
\end{table}



\subsection{General applicability}

We replicated our method on four more cities from the Transportation Networks repository \cite{TransportationNetworks}: Anaheim, Eastern Massachusetts, Berlin T., Sioux Falls, and Barcelona reported above. In all of the cases, the PoA was positive, ranging from 1 in Anaheim to 4 percent in Sioux Falls. After five days of applying Greedy algorithm (Sec. \ref{sec:greedy}) the inequity is reduced below 0.2 of the original value and after 50 days to 0.01 in Anaheim, Berlin and Sioux Falls (\ref{tab:stats_all_cities}. This demonstrates general applicability across the road networks and stability of the obtained performance.


\begin{table}[!ht]
\centering
\begin{tabular}{
    l
    S[table-format=7.0]
    S[table-format=7.0]
    S[table-format=1.2]
    S[table-format=5.2]
    S[table-format=1.2]
    S[table-format=1.2]
    S[table-format=1.2]
    S[table-format=1.2]
    }
\hline
City & {$T_{\text{UE}}$ (mins)} & {$T_{\text{SO}}$ (mins)} & {PoA} & {$\Sigma_{\text{od}}I_1^\text{od}$} & {$\Sigma_{\text{od}}I_5^\text{od}/\Sigma_{\text{od}}I_1^\text{od}$} & {$\Sigma_{\text{od}}I_{10}^\text{od}/\Sigma_{\text{od}}I_1^\text{od}$} & {$\Sigma_{\text{od}}I_{20}^\text{od}/\Sigma_{\text{od}}I_1^\text{od}$} & {$\Sigma_{\text{od}}I_{50}^\text{od}/\Sigma_{\text{od}}I_1^\text{od}$} \\

\hline
Barcelona & 1297794 &1268541 & 1.02 & 2066.71 & 0.16 & 0.07 & 0.06 & 0.04 \\
Anaheim & 1322588 & 1304584 & 1.01 & 3028.83 & 0.15 & 0.05 & 0.02 & 0.01  \\
EMA & 28183 & 27325 & 1.03 & 10.29 & 0.17 & 0.16 & 0.13 & 0.06 \\
Berlin T. & 581509 & 565388 & 1.03& 53890.17 & 0.10 & 0.03 & 0.01 & 0.01 \\
Sioux Falls & 7480157 & 7194761 & 1.04 & 307401.07 & 0.11 & 0.04 & 0.02 & 0.01 \\
\hline
\end{tabular}
\caption{Price of Anarchy in five Transportation Networks repository cities \cite{TransportationNetworks} which is eliminated with Wardropian Cycles. The initial inequity $I_1$ steadily decreases after 5, 10, 20 days of daily assignments and after 50 days ($I_{50}$) the Inequity is largely eliminated.}
\label{tab:stats_all_cities}
\end{table}


\section{Conclusions}
We introduce an idea that can provide a new paradigm of traffic assignment. 
Unlike previous approaches, our solution achieves both optimality and equity.
We show that CAV-operated cities do not have to be dystopian systems of fully compliant individuals, sacrificing for the sake of global optima. 
Quite contrary, with the Wardopian Cycles the new paradigm becomes possible: Cyclical User Equilibrium (CUE), which not only reduces total system costs, but also benefits each participating user.
Similarly, to Nash-like Wardrop User Equilibrium, no user is willing to opt-out from the system in CUE as it would increase his travel time (in a long run).
This is achieved by extending the scope of coordination to a series of days. 
Unlike in classic System Optimum, where drivers would need to coordinate on a single day (to load paths in such a way that total costs are minimised), in CUE the coordination spans in time. 

Our proposed method bridges the gap between User Equilibrium and System Optimum, which is of great importance for the efficiency of our future transport systems. 
Numerical results demonstrate the feasibility of this approach, even in city-scale systems and within short time frames. 
Only in Barcelona we can reduce the total travel time by 962~468 minutes in a single peak-hour, which may be translated to noticeably reduced mileage, fuel consumption, CO2 emission, noise, etc.
This unleashes the potential to remove the price of anarchy from future systems, making it a promising candidate for real-world traffic management. 
Its applicability will start with the rise of CAVs, paving the way for more efficient and fair traffic systems in the coming years.
Arguably, such concept is technically feasible even today, without CAVs, and requires only full compliance routing of informed drivers (which so far was beyond reach of system operators).

Here we provide the conceptual framework with a theoretical background, class of methods, theorems, and proofs. We propose methods to improve the basic cycles by reducing their: length, travel time variance, cumulated disutility, and inter-cycle inequality. We generalise the formal cycles to inequity minimising daily assignment and formulate the concept of Cyclical User Equilibrium. We postulate the new System Optimal assignment, which ensures that for each OD pair travel costs are not worse than in UE. Finally, we show how the concept can be applied for a \emph{bit} of traffic, which, hopefully, will allow future operationalisation of the theoretical concept.

This conceptual paper rather opens than answers research questions. While it shows that systems can be both fair and optimal, it introduces a whole class of new problems: 

\paragraph{Can it be applied for mixed traffic?} When only part of the flow complies, while remaining part follows the classic User Equilibrium? If so, how do individual humans adapt to actions of the fleet and how the fleet anticipates those adaptations?

\paragraph{Will it work for real-world traffic?} Here we assumed a strict relation between path flows, travel times, and static macroscopic model of traffic flow. What if (like in reality) above does not hold true? This calls for broader experimental scheme to verify whether, and to which extent, the vehicles are interchangeable among paths, i.e. if swapping routes of two vehicles (of the same OD pair and similar departure time) changes their and system travel times. 

\paragraph{Is it compatible with pricing schemes?} We argue that monetary pricing of paths contradict the idea of equity behind Wardopian Cycles, since the rich users will buy themselves priority to use faster paths while others would not be able to afford it. However, alternative, non-monetary schemes like Karma \cite{riehl2024karma} or mobility credits \cite{nie2012transaction} seems to be compatible.

\paragraph{What long-term dynamics will it induce?} The road traffic is just an element of the complex system of urban mobility. Introducing more attractive car trips while on one hand may improve the system performance, which, on the other, may induce modal shifts, also from active modes and public transport - not desired in sustainable policies.

\paragraph{Is it applicable for unexpected events?} Disruptions (accidents, demonstrations, maintenance, sport events, etc.) may cause significant delays, unexpected to drivers. This is typically unfair for the unlucky ones. How applying the equity minimising algorithm from Sec. \ref{sec:greedy} may compensate it, remains to be researched.

Apart from those, the concept has not been tested in: competition between platforms (possibly offering own assignments, optimal within their fleet only), when perceived costs of routes are different (and each user ranks paths differently), multi-modal setting (with drivers shifting to public transit and active modes), or in trip-chain context (when travellers perform multiple trips per day).

Fortunately, researchers seem to have some time to work on the above questions before the CAVs will start operating on a big scale and efficiently identify collective routing strategies. If that time is used effectively, society may benefit from an operational assignment scheme making our cities both optimal and fair.




\paragraph{Acknowledgement} This research is financed by the European Union within the Horizon Europe Framework Programme, ERC Starting Grant number 101075838: COeXISTENCE.

\bibliography{reference}

\newpage
\section*{Appendix}
\appendix

\section{Proof of the Claim \ref{cl:np_hard} }
\label{app:NPhard}

In this section, we prove that reduction of a Wardropian Cycle length by the flow partition method is \textit{NP}-hard.

\begin{definition} 
A   \texttt{MEAN-PARTITION} Problem is defined as follows:

     For a given OD pair with $Q$ daily drivers that travel on $K$ routes with a given System Optimal assignment that assigns $Q_k$ agents to the route $r_k$, find a partition of assignment $\mathcal{J} = [J_1, \dots, J_S]$ such that every subset has equal mean travel time.
\end{definition}

\begin{proposition}
    \texttt{MEAN-PARTITION} Problem is \textit{NP}-hard.
   
\end{proposition}

\begin{proof}

Assume that the problem introduced above has a polynomial-time algorithm.

To prove that this problem is \textit{NP}-hard, we need to construct a solution to a problem that is known to be \textit{NP} using problem defined in the proposition. We will use a problem called the \texttt{SUBSET-SUM} problem \cite{garey1979computers}.

\begin{definition}
    \texttt{SUBSET-SUM} problem: Given a finite (multi-) set $A$ of integers of size $n$ and an integer $B$ decide whether there is a subset $A' \subset A$ such that the sum of the elements of $A'$ is exactly $B$.
\end{definition}

Take any instance of a Subset Sum problem such that $B = 0$. W.l.o.g. we can assume that sum of this set is equal 0, in such approach we look for non-trivial solutions to such problem (not the one that takes all elements in the set). Observe that if $A'$ is a proper subset of $A$ that sums up to 0, then its complement $A'^C$ is also such a set. 

Now, pick an arbitrary integer $\hat{t} > |\text{min } a \in A|$ and add it to every element of $A$. We can treat the resulting multiset of positive integers as a vector of travel times on an OD pair of $n$ agents. 

Using the algorithm for solving the \texttt{MEAN-PARTITION} problem, we find the partition of this set into subsets of equal mean.

Any set in this partition, after subtracting $\hat{t}$ from every element, satisfies the \texttt{SUBSET-SUM} problem, hence we constructed a polynomial algorithm for a problem known to be \textit{NP}-complete. Contradiction.

To finish the proof, we need to show that \texttt{MEAN-PARTITION} problem is in the \textit{NP} class. This is indeed true. Assume that we were given a supposed answer to this problem. We can decide in linear time (summing the elements in the partitions and calculating the averages) whether this solution is valid or not, hence \texttt{MEAN-PARTITION} $\in$ \textit{NP}.

\end{proof}

\section{Proof of Theorem \ref{prop:greedy}}
\label{app:algo_properties}

In this section, we prove that the assignment constructed by the Greedy algorithm yields the minimal system inequity on the next day out of all possible assignments.

\begin{proof}
    Let $\Gamma^G$ be a greedy assignment rule, and $\Gamma$ any other assignment rule.


    On day $J$, $\Gamma^G$ assigns to the sorted deviations from the previous day $\mathcal{D} = ([D_J]_{i_1} \geq [D_J]_{i_2} \geq \dots \geq [D_J]_{i_Q})$ travel times sorted in reverse order $(t_1 \leq t_2 \leq \dots \leq t_Q)$. We can express the agents' inequities on the next day as the following: 
    \begin{equation}
        \textbf{I}_{J+1}(\Gamma^G) = \sum_{j=1}^{Q} ([D_J]_{i_j} + t_j)^2 = \sum_{j=1}^{Q} ([D_J]_{i_j})^2 + t_j^2  + 2\sum_{j=1}^{Q} [D_J]_{i_j} \cdot t_j
    \end{equation}
    We can express the assignment given by $\Gamma$ as some permutation $\sigma$ of $(t_1 \leq t_2 \leq \dots \leq t_Q)$: $(t_{\sigma(1)} \leq t_{\sigma(2)} \leq \dots \leq t_{\sigma(Q)})$. Then the inequity can be written as follows:
     \begin{equation}
        \textbf{I}_{J+1}(\Gamma) = \sum_{j=1}^{Q} ([D_J]_{i_j} + t_{\sigma(j)})^2 = \sum_{j=1}^{Q} ([D_J]_{i_j})^2 + t_{\sigma(j)}^2  + 2\sum_{j=1}^{Q} [D_J]_{i_j} \cdot t_{\sigma(j)}
    \end{equation}

    The only difference in the inequity formulas between $\Gamma^G$ and $\Gamma$ are the factors $2\sum_{j=1}^{Q} [D_J]_{i_j} \cdot t_j$ and  $2\sum_{j=1}^{Q} [D_J]_{i_j} \cdot t_{\sigma(j)}$ respectively. But from the \textit{Rearrangement Inequality} (see \cite{hardy1952inequalities}), we know that:
    \begin{equation*}
        \sum_{j=1}^{Q} [D_J]_{i_j} \cdot t_j \leq \sum_{j=1}^{Q} [D_J]_{i_j} \cdot t_{\sigma(j)}
    \end{equation*}
    For any permutation $\sigma$, which completes the proof.
\end{proof}

\section{optimising the cycle length - technical results}
\begin{proposition}[optimisation based on a single cycle is NP-hard] 
\label{Prop_cycleNPhard}
Finding an optimal permutation in \eqref{eq:opt_single} is NP-hard.
\end{proposition}
\begin{proof}    
Suppose there is an algorithm solving \eqref{eq:opt_single}. Then, noting that $\max_n |t_{n} - \hat{t}|$ is a lower bound in \eqref{eq:opt_single}, one can use this algorithm, at a linear additional cost, to solve the decision problem: 
\newline \noindent \emph {Does there exist a permutation $\sigma \in S_N$ such that  
\begin{equation}
    \label{eq:decision_single}
    \max_{N_S, N_L \in \{1,\dots, N\}} \left|\sum_{n=N_S}^{N_S + N_L} (t_{\sigma(n)} - \hat{t}) \right| = \max_n |t_{n} - \hat{t}|
\end{equation}
is satisfied?}

Now, consider any finite sequence $t_1, t_2, \dots, t_{N-2}$ of negative integers. Let $\bar{t} = \left|\sum_{n=1}^{N-2} t_n\right|$ and assume without loss of generality that $|t_n|< \bar{t}/2$ for every $n = 1,\dots, N-2$. Let $t_{N} = t_{N-1} = \bar{t}/2$. Then $t_1 + t_2 + \dots + t_N = 0$ and so the algorithm solving \eqref{eq:decision_single}  can decide whether there is a permutation $\sigma$ such that $\left|\sum_{n=N_S}^{N_S + N_L} t_{\sigma(n)}\right| \le \bar{t}/2$ for every $N_S, N_L$. Any such permutation, however, has to be, up to a cyclical shift, of the form 
\begin{equation*}
    \sigma = (N-1, n_1, \dots, n_k , N , n_{k+1}, \dots, n_{N-2} )
\end{equation*}
where $(n_1, n_2, \dots, n_{N-2})$ is a permutation of the set $\{1,2,\dots, N-2\}$ and  $t_{n_1} + \dots + t_{n_k} = -\bar{t}/2$ and $t_{n_{k+1}} + \dots + t_{n_{N-2}} = -\bar{t}/2$,
for some $k$. Therefore, an algorithm solving \eqref{eq:opt_single} can be used to decide, at a linear additional cost, whether numbers $t_1, t_2, \dots, t_{N-2}$ can be split into two subsets with equal totals, which is an instance of the NP-complete problem \texttt{SUBSET-SUM}.
\end{proof}

The general optimisation problem, where the assignments of drivers to routes do not follow a single permutation (cycle), yet after $N$ days every driver has used all the routes the required number of times is stated below. 

\begin{definition}
Permutations $\sigma_1, \dots, \sigma_N \in S_N$ are \emph{compatible} if $(\sigma_1(n), \sigma_2(n), \dots \sigma_N(n)) \in S_N$ for every $n = 1,\dots, N$.  
\end{definition}
We note that if $\sigma \in S_N$ then the permutations $\sigma_i(n):=\sigma(n+i)$, for $i=1,\dots,N$, i.e. permutations resulting from a cycle, are compatible and so a cyclical assignment based on a single permutation is encompassed by the general framework. 

\begin{definition}[Cycle inequity (general) optimisation problem]
\label{def:opt_multi}
Let $t_1, t_2, \dots, t_N$ be a sequence of positive natural numbers and let $\hat{t} := \frac 1 N \sum_{n=1}^N t_n$. The \emph{cycle inequity optimisation problem} is given by:
\begin{equation}
\label{eq:opt_multi}
    \min_{compatible \, \sigma_1, \sigma_2, \dots, \sigma_N \in S_N}  \left\{ \max_{\nu, N_L \in \{1,\dots, N\}} \left|\sum_{n=1}^{N_L} (t_{\sigma_\nu(n)} - \hat{t}) \right| \right\}.
\end{equation}
\end{definition}

\begin{conjecture}
    The problem \eqref{eq:opt_multi} is NP-hard.
\end{conjecture}

\begin{example}[Inequity minimisation based on a single permutation is inefficient]
\label{ex:multi_better}
Consider $N=6$ and $(t_1, t_2, t_3, t_4, t_5, t_6) = (4,4,1,-3,-3,-3)$. Then $\hat{t} = 0$. Consider compatible permutations $\sigma_1, \dots, \sigma_6$ resulting in the following matrix (where $t_{\sigma}:= [t_{\sigma(1)}, \dots, t_{\sigma(N)}])$
\begin{equation*}
\begin{bmatrix}
    t_{\sigma_1}\\
    t_{\sigma_2}\\
    t_{\sigma_3}\\
    t_{\sigma_4}\\
    t_{\sigma_5}\\
    t_{\sigma_6}
\end{bmatrix}
= 
\begin{bmatrix}
    4 & -3 & -3 & 4 & -3 & 1\\
    4 & -3 & -3 & 4 & 1 & -3\\
    -3 & 4 & -3 & 1 & 4 & -3\\
    -3 & 4 & 1 & -3 & 4 & -3\\
    -3 & 1 & 4 & -3 & -3 & 4\\
    1 & -3 & 4 & -3 & -3 & 4\\
\end{bmatrix}
\end{equation*}
One can easily verify that the maximum cumulative deviation as per eq. \eqref{eq:opt_multi} is equal to $4$ and so it is optimal. On the other hand, the single permutation minimisation problem \eqref{eq:opt_single} results in a permutation (non-unique, up to cyclical shifts) $\sigma = (1,4,2,5,3,6)$ with $t_\sigma = (4,-3,4,-3,1,-3)$ for which the maximum cumulative deviation in eq. \eqref{eq:opt_single} equals $5$.
\end{example}

\section{Further measures of compliance}
In this section, we propose some other common-sense measures of discontent, which, however, need to be studied behaviourally in order to propose a general minimisation objective which would increase compliance of potential Wardropian-Cycle users. 

\begin{definition}[Max period above-average travel]
We define the \emph{maximum period of above-average} travel  
\begin{equation*}
    M^i = \max \{j_1: \exists j_0 :  [D_j]_i > 0 \mbox{ for every } j = j_0, \dots, j_0+j_1 \}.
\end{equation*}
\end{definition}

\begin{definition}[One-sided discontent]
For every day $L$ of the routing scheme and every driver $i$ subscribed to this routing scheme, we define the \emph{cumulative discontent} as the prefix sum of the sequence of cumulative travel time deviations (i.e. individual discontents $\mathcal{I}^i_J$).  
\begin{equation*}
\mathcal{D}^{+i}_L := \sum_{J=1}^L \mathcal{I}^{+i}_J =  \sum_{J=1}^L \sum_{j=1}^J \max \left\{0, [D_j]_i\right\}.
\end{equation*}
\end{definition}

\begin{example}[Illustrating the importance of using $\mathcal{H}$ while constructing the Cycles]
\label{ex:discontent}
Suppose the Wardrop Cycle consists of $3$ days of travel via three routes whose travel times are $(6, 1, 2)$. The deviations from the average, which equals $3$, are given by $(3,-2,-1)$. Then the matrices of deviations, cumulative deviations $\mathcal{I}^i_J$ and cumulative discontent of the scheme defined by permutation $\sigma_1 = (1,2,3)$ is given by (rows for different agents)  
$$ \mathcal{D}_1 (\sigma_1)= 
\begin{bmatrix}
    3 \\
    -2\\
    -1
\end{bmatrix},
\mathcal{D}_2 (\sigma_1)= 
\begin{bmatrix}
    1\\
    -3\\
    -2
\end{bmatrix}, 
\mathcal{D}_3 (\sigma_1)= 
\begin{bmatrix}
    0\\
    0\\
    0
\end{bmatrix}, 
\mathcal{H}_3 (\sigma_1)= 
\begin{bmatrix}
    3 & 4 & 4\\
    -2 & -5 & -5\\
    -1 & 1 & 1
\end{bmatrix}.
$$
On the other hand, the alternative permutation $\sigma_2 = (1,3,2)$ generates deviations and cumulative discontent 
$$ \mathcal{D}_1 (\sigma_2)= 
\begin{bmatrix}
    3 \\
    -2\\
    -1
\end{bmatrix},
\mathcal{D}_2 (\sigma_2)= 
\begin{bmatrix}
    2\\
    1\\
    -3
\end{bmatrix}, 
\mathcal{D}_3 (\sigma_2)= 
\begin{bmatrix}
    0\\
    0\\
    0
\end{bmatrix}, 
\mathcal{H}_3 (\sigma_2)= 
\begin{bmatrix}
    3 & 5 & 5\\
    -1 & -4 & -4\\
    -2 & -1 & -1
\end{bmatrix}.
$$
Comparison of the matrices reveals that the maximum cumulative discontent for $\sigma_1$ equals $4$, however it is $5$ for $\sigma_2$. Therefore, even though the inequity, as defined in definition \ref{def:opt_single}, on every day $1,2,3$ is exactly the same (as the individual cumulative deviations form the same multisets: $\{-1,-2,3\}$ on day 1, $\{1,2, -3\}$ on day $2$ and $\{0,0,0\}$ on day $3$), the maximum cumulative discontent is worse for $\sigma_2$. Therefore, the scheme defined by $\sigma_1$ seems to be preferred.
\end{example}

\end{document}